\documentclass[journal,10pt,doublecolumn]{IEEEtran}

\IEEEoverridecommandlockouts
\usepackage{cite}
\usepackage{url}
\usepackage{balance}
\usepackage{amsmath,amssymb,amsfonts}
\usepackage{algorithmic}
\usepackage{graphicx}
\usepackage{textcomp}
\usepackage{xcolor}
\usepackage{booktabs}
\usepackage[outdir=./]{epstopdf}
\usepackage{flushend}
\def\BibTeX{{\rm B\kern-.05em{\sc i\kern-.025em b}\kern-.08em
    T\kern-.1667em\lower.7ex\hbox{E}\kern-.125em}}
\usepackage{comment}

\usepackage{amsmath,amssymb,bm,dsfont}
\usepackage{amsthm,nccmath,lipsum,cuted}
\newtheorem{theorem}{Theorem}
\newtheorem{proposition}{Proposition}
\usepackage{mathtools}

\newcommand{\symtext}[2]{\ensuremath{\stackrel{\text{#1}}{#2}}}

\DeclareMathOperator*{\argmax}{arg\,max}
\DeclareMathOperator*{\argmin}{arg\,min}

\def\P{{\mathbb P}}  %%%   appears in many equations  Prob
\def\E{{\mathbb E}}  %%%   for "expectation value" (in math mode)
\def\S{{S}}
\def\A{{A}}

\def\D{{D}}
\def\B{{\mathcal B}}

\graphicspath{ {./resources/figures/} }

\newcommand{\review}[1]{{#1}}
\newcommand{\note}[1]{{\color{blue}#1}}
\renewcommand{\note}[1]{#1}

\newcounter{storeeqcounter}
\newcounter{tempeqcounter}

\definecolor{applegreen}{rgb}{0.55, 0.71, 0.0}
\definecolor{bronze}{rgb}{0.8, 0.5, 0.2}
\definecolor{coolblack}{rgb}{0.0, 0.18, 0.39}

\usepackage{url}
\newcounter{todo}

\makeatletter

\newcommand\listtodoname{List of todos}
\newcommand\listoftodos{%
	\section*{\listtodoname}\@starttoc{tod}}
\makeatother

\begin{document}

\title{\note{Scheduling of Wireless Edge Networks for Feedback-Based Interactive Applications}}

\author{\IEEEauthorblockN{
		Samuele Zoppi\IEEEauthorrefmark{1},
		Jaya Prakash Champati\IEEEauthorrefmark{2},
		James Gross\IEEEauthorrefmark{3}, 
		and Wolfgang Kellerer\IEEEauthorrefmark{1}\\
	}
	\IEEEauthorblockA{
		\IEEEauthorrefmark{1}Chair of Communication Networks, Technical University of Munich, Germany\\
		\IEEEauthorrefmark{2}IMDEA Networks Institute, Spain\\
		\IEEEauthorrefmark{3}School of Electrical Engineering and Computer Science, KTH Royal Institute of Technology, Sweden\\
		Email: \{samuele.zoppi, wolfgang.kellerer\}\kern-1pt@tum.de, jamesgr\kern-1pt@kth.se, jaya.champati\kern-1pt@imdea.org\\
		\vspace{-2ex}
	}
}

% ********** Todo list **********
% \listoftodos

% \cleardoublepage
% ********** Todo list **********

\maketitle

\begin{abstract}
\note{
Interactive applications with automated feedback will largely influence the design of future networked infrastructures. In such applications, status information about an environment of interest is captured and forwarded to a compute node, which analyzes the information and generates a feedback message. Timely processing and forwarding must ensure the feedback information to be still applicable; thus, the quality-of-service parameter for such applications is the end-to-end latency over the entire loop. By modelling the communication of a feedback loop as a two-hop network, we address the problem of allocating network resources in order to minimize the delay violation probability (DVP), i.e. the probability of the end-to-end latency exceeding a target value. We investigate the influence of the network queue states along the network path on the performance of semi-static and dynamic scheduling policies. The former determine the schedule prior to the transmission of the packet, while the latter benefit from feedback on the queue states as time evolves and reallocate time slots depending on the queue’s evolution. The performance of the proposed policies is evaluated for variations in several system parameters and comparison baselines. Results show that the proposed semi-static policy achieves close-to-optimal DVP and the dynamic policy outperforms the state-of-the-art algorithms. 
% The development of wireless communications stimulates novel feedback applications, such as Cyber-Physical Systems and Tactile Internet, to operate over wireless links.
% Feedback applications pose stringent delay and reliability requirements to the transmission of packets over both sensor-to-controller and controller-to-actuator links.
% %Transmissions over wireless links, however, causes packets to be randomly lost.

% With the goal of finding scheduling policies to minimize DVP, we capture the correlation between the resource allocations of the links and exploit it to compute network schedules.
% The availability of network state information at the scheduler can significantly improve DVP. 
% Existing methods, however, do not exploit it due to the large overhead required to acquire network information.
% We investigate how network state information can be used to compute offline and online scheduling policies in a two-hops network.
% We propose offline scheduling policies based on upper bounds on DVP, and, noting that DVP cannot be directly used for online scheduling, we develop an online scheduling policy relating DVP to the network's throughput.
% Via simulations, we demonstrate that offline scheduling policies achieve close-to-optimal performances and that throughput-optimal online policies achieve lower DVP in comparison to the classical Max Weight, Weighted-Fair Queuing, and Backpressure algorithms.
%We additionally show that exploiting network state information leads to considerable improvements in the DVP of time-critical arrivals.
}
\end{abstract}

% When considering individual time-critical arrivals, the availability of network state information at the scheduler can significantly improve the achievable performance. 
% In the literature, this is exploited by backpressure algorithms~\cite{Tassiulas1990StabilityNetworks}. Their application to large-scale networks, however, is limited by the required overhead of acquiring network state information.
% This is particularly crucial for delay-constrained applications, where the time required to acquire the network state can jeopardize its benefits~\cite{Singh2019ThroughputLinks}.
% To overcome this limitation, we investigate the benefit of using network state information in a two-hops network representing the sensor-to-controller and controller-to-actuator links of feedback-applications.
% In particular, we study the impact of network state information for different scheduling methods, in order to apply them to a broad range of wireless networks and applications.
% On one hand, \emph{offline} scheduling exclusively relies on the initial network state and is suitable for resource-limited wireless networks, such as Wireless Sensor Networks (WSN). 
% On the other hand, \emph{online} scheduling benefits from network state availability at each frame and is suitable for infrastructure-based wireless networks, such as cellular networks, which take advantage of centralized network coordination and reliable feedback channels.


\section{Introduction}
\note{
% Introduce the interactive applications, what types of such interactions are there, how is their principle working, and what makes these applications interesting
Interactive applications with automated feedback are arguably one of the most discussed application class when it comes to large-scale impact in future networked infrastructures today~\cite{Saad2020AProblems}.
In the literature, we can broadly distinguish two sub-cases of this new application class, namely cyber-physical systems (CPS) and human-in-the-loop systems~\cite{Rajkumar2010Cyber-physicalRevolution, Schirner2013TheSystems, 10.1145/2594368.2594383}.
In both cases, however, the underlying principle is the same: status information about a plant or an environment of interest is captured, and forwarded to a compute node, where the information is analyzed and potentially a feedback is generated in form of an actuation command or an augmentation/perceptual feedback; see Fig.~\ref{fig:network-architecture}.
In CPS, we encounter applications where direct actuation is applied to a physical object, for instance, in the context of industrial automation and automated driving. 
In contrast, augmented reality, cognitive assistance, and also to some extent virtual reality fall under the category of human-in-the-loop systems, where no direct actuation results from the feedback; instead, a human is presented perceptual feedback which potentially triggers some human reaction~\cite{10.1145/2594368.2594383}.

% Define the QoS requirements of these applications, make the reader aware of the "end-to-end latency" (side note to discuss: the way we formulate it, it is equivalent to transient AoI), discuss the value settings for different application types (motion-to-photon latency, mention also perhaps that for control applications the latency goal could be variable over time). Discuss the implications of not meeting the required latency.
Key features about these applications is their tight integration with reality as well as the degree of automation that they allow.
Furthermore, by offloading these applications to networked infrastructures, an almost ubiquitous availability of corresponding services will be enabled in the future. 
However, the successful deployment of such applications rests crucially on a timely processing and forwarding along the pipeline to ensure that the feedback information is timely with respect to the original sensing data.
Thus, the quality-of-service (QoS) parameter for interactive applications is the latency over the entire loop, i.e. from capturing the status information until the point in time when the corresponding feedback information is exposed~\cite{Zhang2013a}.
For instance, motion-to-photon latency in virtual reality is a well-known concept that captures this QoS parameter.
Loosely speaking, one might refer to this QoS parameter as the \textit{end-to-end latency} with the crucial differentiation that a flow conversion occurs at the point of computation.
In fact, depending on the application, for a certain fraction of sensing data no feedback is generated at all, for instance in the case of cognitive assistance.
Finally, different applications may also require time-varying end-to-end latencies.
For instance, in CPS it is known that the criticality of sensing information can vary, leading to time-varying end-to-end latency constraints as the plant dynamics evolve.
\begin{figure}[t]
	\centering
	\includegraphics{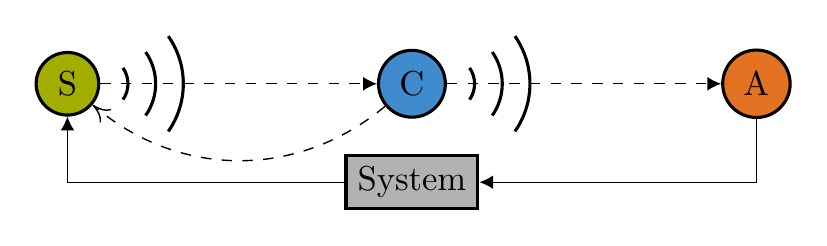}
	\caption{Architecture of interactive applications with automated feedback. Sensor (S), controller (C), and actuator (A) operate in a closed loop via a two-hop wireless network path.}
	\label{fig:network-architecture}
\end{figure}

% Now start a discussion on the optimal support of such applications. For several reasons wireless plays an important role for such applications. In addition, compute resources are needed somewhere in the network. This all introduces sources of random processing and losses, for which we need to provide extra resources for compensation. But how to steer this provisioning of extra resources? Try to optimize for the DVP, so the likelihood of keeping the deadline ...
Given the relevance of interactive applications, as well as their novel QoS requirements, a central question relates to the optimal support of such applications by networked infrastructures.
Truly ubiquitous service offering mandates wireless connectivity to the point of computation.
Furthermore, general latency requirements mandate near-by computational service, typically provided by edge computing~\cite{Shi2016EdgeChallenges}.
Hence, a networked infrastructure realizing an interactive application needs to provide bounded end-to-end latencies over a concatenated, heterogeneous network path with at least one compute element incorporated~\cite{Saad2020AProblems}. 
This complicates the provisioning of end-to-end latencies, as the individual elements of the network path are typically subject to multiple random effects, such as fading in case of the wireless links, operating system scheduling effects in case of the compute backend, and/or cross-traffic for both the communication and computation elements~\cite{Champati2020TransientRouting}. 
As a consequence, the end-to-end latency becomes essentially a random variable, which also depends on parameterizations of the network path such as resource allocation, task prioritization etc.
In order to steer the parameterization of the network path facing the uncertainty from different effects, a suitable metric is to minimize the delay violation probability (DVP), defined as the probability of the end-to-end latency exceeding a (constant or varying) target value. 
The key question then is to manage network path resources to minimize the DVP.
Addressing this question is at the heart of this work.

% Now comes the main twist of the paper: We think that the queue state should go into this optimization, so we need a DVP optimization over the entire network path. This can be done in two fashions, online and offline. Explain that a bit. The key question of this paper is what the consequences of this approach are mathematically/in terms of complexity, as well as in terms of performance.
Acknowledging the fact that latency targets might vary over time, as well as cross-traffic contributing to the utilization of the individual elements of the path, we are interested in the fundamental question if \textit{initial conditions} of the network path should influence the path parameterization.
By initial conditions, we refer to the queue states along the network path, which is modelled as a concatenated queuing network.
This might be included in at least two different ways.
On the one hand, once a new sensor reading is available, a \emph{semi-static} resource allocation might be determined which is kept during the subsequent evolution of the system until the corresponding actuation command is delivered to the actuator.
On the other hand, a \emph{dynamic} policy constantly adjusts the resource allocation during the evolution of the system until the actuation command is delivered. 
Obviously, a constant adaptation of the resource allocation should provide a better performance in terms of observed delay violations at the price of increased signaling load.
But exactly how such algorithms should work, which complexity they bring, and which performance differences they imply for interactive applications is to the best of our knowledge open to date.

% Now discuss novelty of this work (i.e. what differentiates it from other approaches that optimize for latency, or that take queue states into account)?
% And then list the contributions specifically.
In this paper, the above questions are investigated for a two-hop network path incorporating the loop communication of the status information to a compute node and of the feedback message to an actuator, cf. Fig.~\ref{fig:network-architecture}.
The two-hop system follows an uplink/downlink model where network resources need to be assigned \emph{in competition}.
This applies to communication systems where time resources are typically shared between uplink and downlink, such as WirelessHART, LTE TDD, and 5G TDD NR.
The end-to-end latency of packets is dominated by the random delay caused by the retransmission of lost packets and thus the processing delay introduced by the compute node is assumed to be negligible.
%Packet transmission at each link experiences a random delay caused by the retransmission of lost packets and
Packet transmission follows a time-slotted medium access where network resources are organized in frames.
In each frame, the available time slots are entirely allocated to the two links that compete for resources.
Therefore, given the initial queue states of the two links, we investigate scheduling policies that allocate time slots of each frame to the links in order to minimize the DVP of packets belonging to interactive applications.
%In this work, this is achieved by characterizing the DVP of packets traversing a two-hop network path and deriving semi-static and dynamic resource allocation mechanisms that minimize DVP, which lead to the following contributions:

The main contributions of this paper are summarized in the following:
\begin{itemize}
    \item We show that the closed-form expression of DVP is intractable and derive two upper bounds for the DVP of packets traversing a two-hop network path given the initial network conditions.
    \item Novel heuristic scheduling policies that compute a semi-static resource allocation are proposed.
    %\item Proving that WTB is convex we provide an efficient way of computing offline scheduling policies.
    \item Noting that DVP cannot be directly used for dynamic resource allocation, a dynamic heuristic scheduling policy that maximizes the network's throughput is proposed.
    % \item The proposed online scheduling method maximizes the expected throughput and it is computed by solving a finite-horizon Markov Decision Process (MDP). % and can be implemented in any centralized network entity.
    \item A simulation study of semi-static scheduling policies as well as a comparison between the dynamic scheduling policy, the classical Backpressure (BP)~\cite{Tassiulas1990StabilityNetworks}, Max Weight (MW)~\cite{neely2010stochastic}, and Weighted-Fair Queuing (WFQ)~\cite{parekh1993generalized} scheduling policies is presented.
\end{itemize}
% The simulations demonstrate that the proposed offline scheduling policies achieve close-to-optimal performances and the online scheduling policy achieves lower DVP in comparison to the classical algorithms.
% Furthermore, we analyze and discuss the impact of network state information and of different system parameters on the achievable performance of the proposed schemes.    

The rest of the paper is structured as follows.
Sec.~\ref{sec:related-work} provides a discussion of the related work.
Sec.~\ref{sec:network-model} defines the model of the two-hop network path and the problem statement.
Sec.~\ref{sec:dvp-heuristic} provides a general derivation of DVP and discusses its application for scheduling. 
In Sec.~\ref{sec:offline-scheduling} heuristic semi-static scheduling policies are derived, while Sec.~\ref{sec:online-scheduling} describes an MDP-based heuristic dynamic scheduler.
Sec.~\ref{sec:results} evaluates the performance of the proposed scheduling methods and provides a discussion on their applicability in different scenarios.
Finally, we conclude in Sec.~\ref{sec:conclusions}.

\section{Related Work}\label{sec:related-work}
Several existing works tackle the problem of resource allocation in a multi-hop wireless network to support time-critical applications.
% Network state information
Methods that make use of the queue state to allocate network resources follow a \emph{theoretical} approach~\cite{Tassiulas1993DynamicConnectivity, Tassiulas1990StabilityNetworks, Singh2019ThroughputLinks}.
In their pioneering work~\cite{Tassiulas1993DynamicConnectivity}, Tassiulas et al. derived the Max Weight scheduling policy, which allocates resources based on the transmitters' backlogs and achieves maximum throughput and minimum delay.
Their scenario, however, is different from the one in this work as they only considered a single-hop network.
For a multi-hop network, maximum throughput was achieved by the backpressure algorithm~\cite{Tassiulas1990StabilityNetworks}, which allocates resources based on the backpressure of queues in the network.
Differently than this work that minimizes the DVP of time-critical packets, their scheduling policy focused on throughput optimality.
Singh et al.~\cite{Singh2019ThroughputLinks} exploited information about the queue state to schedule transmissions in order to maximize throughput under delay constraints.
Their approach, however, considers the steady-state performance of packets and, differently than our approach, does not allocate resources to optimize the network for a single time-critical arrival.

% 2. Deadline-constrained
From a different perspective, many practical works investigate resource allocation methods for real-time flows in Industrial Wireless Sensor Networks (IWSN).
Some of them compute schedules to allow several time-constrained applications to meet their deadlines assuming \emph{deterministic transmission outcomes}~\cite{Saifullah2010Real-timeNetworks, Saifullah2011End-to-endNetworks,Saifullah2012AccountingReport, Saifullah2015End-to-endNetworks, Wang2017, Modekurthy2019DistributedHART:Networks}.
Saifullah et al.~\cite{Saifullah2010Real-timeNetworks, Saifullah2011End-to-endNetworks} investigated the problem of real-time scheduling subject to end-to-end deadlines between sensors and actuators.
Differently than our work, however, communication between sensors and actuators is deterministic and packet loss is not considered.
In their recent works~\cite{Saifullah2012AccountingReport, Saifullah2015End-to-endNetworks}, communication failures due to packet loss are considered and retransmissions are used. %and sensor-controller and controller-actuator paths are modelled.
In contrast to our work, however, the only provide a deterministic delay model for the communication between sensors, controller, and actuators.
%--> this guy deterministic model too!
Similarly, Wang et al.~\cite{Wang2017} calculated schedules to ensure the worst-case delay of packets in a flow, however, without assuming random packet loss.
Modekurthy et al.~\cite{Modekurthy2019DistributedHART:Networks} derived a distributed deadline-based scheduling algorithm based on the Earliest Deadline First policy.
Also in this case, differently than our scenario, packet loss is not considered and the random end-to-end delay of packets in the network is not characterized.
%Worst-case end-to-end delay analysis of the schedules was provided, however, without re-transmitting lost packets. 

% 3. Reliability-based schedulers
Other IWSN works tackle the problem of reliable communication in presence of \emph{random packet loss}~\cite{Yan2014Hypergraph-basedConstraints, Dobslaw2016, Gaillard2016}.
Dobslaw et al.~\cite{Dobslaw2016}, allocated resources to each transmitter in a path based on the required number of retransmissions to fulfil a given reliability constraint.
Differently than our model, however, their work did not consider a deadline for the packets.
Following a similar approach, Gaillard et al.~\cite{Gaillard2016} extended the pioneer traffic-aware centralized scheduler TASA~\cite{Palattella2012TrafficNetworks} including retransmissions to guarantee flow reliability requirements.
Also in this case, however, traffic is not time-constrained.
Yan et al.~\cite{Yan2014Hypergraph-basedConstraints} developed a scheduling method that allocates time slots to the transmitters in order to maximize the network reliability under delay constraints.
A major difference of their work, which is common to Dobslaw and Gaillard et al., is that reliability constraints are defined for all the flows in the network.
Our approach instead optimizes the network resources to maximize the application reliability of each time-critical packet.

% WSN3 Per-packets methods
Recent works tackle investigate resource allocation methods providing \emph{per-packet} delay and reliability performance ~\cite{Chen2018ProbabilisticControl, Brummet2018ANetworks, Gong2019ReliableNetworks, Soldati2010OptimalNetworks}.
Similarly to DVP, Chen et al.~\cite{Chen2018ProbabilisticControl} computed, for each packet, the number of transmissions required to fulfil the application deadline with a given probability.
Their work, however, only considered a single-hop scenario and cannot be applied to the considered two-hop network path.
Brummet et al.~\cite{Brummet2018ANetworks} developed a method to dynamically allocate retransmissions to each packet and at each network hop subject to delay and reliability requirements.
They followed a different approach as their schedules are designed to fulfil the requirements and limit the maximum number of retransmissions.
Therefore, their scenario did not optimize network resources to maximize the per-packet DVP. %probability of delivery within the deadline.
A similar approach was used by Gong et al.~\cite{Gong2019ReliableNetworks}, which allocated time slots to transmitters fulfilling per-packet delay and reliability constraints while minimizing the number of resources.
Also in this case, however, they considered a finite number of retransmissions and the network resources are not optimized to minimize DVP.
On the contrary, Soldati et al.~\cite{Soldati2010OptimalNetworks} allocated time slots over multiple hops maximizing the end-to-end reliability of each time-critical packet subject to a deadline.
Differently than our approach that characterizes the distribution of end-to-end delay considering the correlation of transmission outcomes over consecutive transmitters, their scenario assumes that resource allocations of consequent transmitters are independent.

This work extends our previous work~\cite{Zoppi2020dynamic} by deriving semi-static scheduling policies and investigating the impact of queue state information on the minimum achievable DVP.
We achieve this thanks to a transient queuing model of the network and by modelling the end-to-end delay of each packet of interactive applications.
This is different from the available state-of-the-art for multiple reasons.
Instead of considering QoS for a stationary flow, we analyse the random end-to-end delay incurred for each time-critical arrival.
The considered queuing model allows us to derive scheduling policies taking into account the correlation between subsequent transmitters introduced by random transmission outcomes.
This is different from the related work as existing scheduling policies optimize network resources based on the interaction of multiple independent flows sharing the network.
Furthermore, by allocating a finite amount of retransmissions, all the existing methods allow application packets to be dropped, which may result in critical failures of feedback systems.
Finally, we investigate the impact of queue state information on semi-static and dynamic scheduling policies, which, to the best of the knowledge of the authors, was never tackled in the literature for a two-hop network path.
}
\section{System Model and Problem Statement}\label{sec:network-model}
%In this section, we formulate the problem of QoS provisioning of a time-critical packet and provide the necessary details of the transient model of the network elements.
\note{We study the communication scheduling problem of a feedback system consisting of a sensor, a control logic, and an actuator. 
The two-hop data communication from the sensor to  the controller, and the controller to the actuator is performed via error-prone wireless links.
We now describe the details of the model and the network elements, and present the formulation of the scheduling problem.}

\subsection{Arrivals, Backlogs, and Departures}
We model the sensor-controller link, and controller-actuator link using a packet-flow, discrete-time, two-queue \review{lossy wireless} network with first-come-first-serve discipline, cf. Fig.~\ref{fig:network-model}.
The time is discretized into slots, which are grouped in frames. 
A sequence of $y$ packets arrives at the first queue in frame $0$\footnote{We consider frame $0$ for notational simplicity; nevertheless, our analysis is equally valid starting the system with any other frame.}. 
These packets are time critical with a requirement that they depart the second queue within next $w$ frames, where $w$ is finite. \note{These packets, for instance, may belong to a time-critical message whose latency could significantly impact the application performance, i.e. the stability/safety of a feedback system.} 
We are thus interested in analyzing the two-queue network path for the time frames $k \in \{0,1,\ldots, w-1\}$. 
In this transient regime, the delay incurred by the time-critical packets depends on the initial backlogs in the queues at \note{frame $0$}, and the temporal variations in the service received by the queues. 
%We are interested in studying the delay characteristic of a sequence of time-critical packets consisting of  that arrive in a .
\begin{figure}[t]
	\centering
	\includegraphics{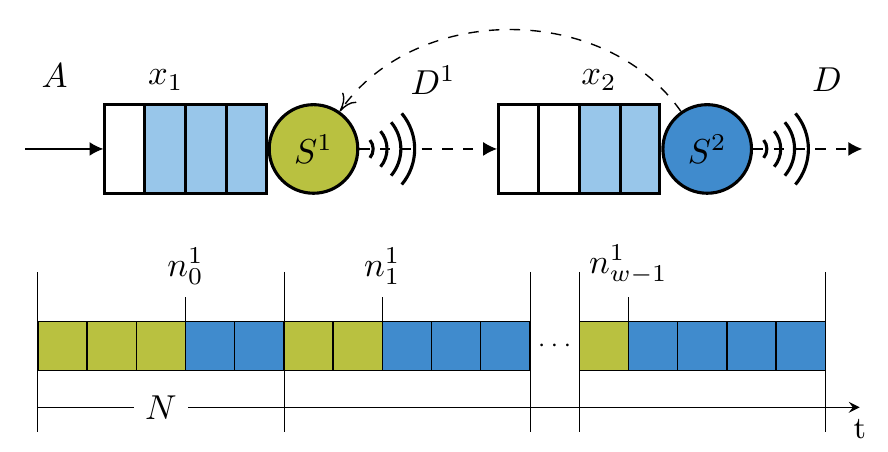}
	\caption{\note{Network model of the two-hop network path. The available time slots are entirely allocated to the two queues at each frame until the deadline.}}
	\label{fig:network-model}
\end{figure}
%The time epochs belong to the set $\{0,1,\ldots,w\}$ for $w < \infty$, and we use $k$ to index them. 
 
%\sz{My version (feel free to use or drop): We are interested in investigating the QoS performance of a single urgent transmission of the NCS.
%In general, this can consist of a train of packets and we model it as the arrival of $y$ time-critical packets in the network at frame $k=0$.}
%To perform this analysis, we need to model the transient behaviour of the network. 
%In this condition, queues have not reached the steady state performance, and initial buffer is present in both queues. 

%The visual representation of the transient model of the network is shown in Fig.~\ref{fig:network-model}.
We use $i\in \{1,2\}$ to index the queues. Let $x_i$ denote the backlog in queue $i$ in frame $0$. 
%The presence of initial backlogs $x_1,x_2$ is modelled as cross traffic arrivals at frame $k=0$.
%We use $i$ to index the queues and model the cumulative traffic arrival as
%\begin{align}\label{eq:cumulative-backlogs}
%\C^i(t)= \sum_{j=0}^{t-1} c^i_j = c^i_0 = x_i, \quad i=1,2.		
%\end{align}	
Let $\A^i(k)$ and $\D^i(k)$ denote the cumulative arrivals and departures at queue $i$, in frame $k$. For $k=0$, all the quantities are set to zero. For $k \geq 1$, we define
\begin{align}
A^\text{1} (k) &= y + x_1,\label{eq:cumulative-arrivals-q1}\\
A^\text{2} (k) &= D^\text{1}(k-1) + x_2, \label{eq:cumulative-arrivals-q2}\\
\D^i(k) &= \sum_{j=0}^{k-1} d^i_j,
\end{align}
where $d^i_j$ is the number of packets departed queue $i$ in frame $j$. 
In Eq.~\eqref{eq:cumulative-arrivals-q2} a one-step delay is introduced between the reception of a packet and its service at the second queue indicating that packets must be fully received before being relayed. 
In the following, we use $A(k) = A^\text{1} (k)$ and $D(k) = D^\text{2}(k)$.
\note{For analytical simplicity, we assume that a packet received by the controller is processed within the same frame of reception, i.e. processing latencies are negligible, and results in a new packet carrying the feedback information.
Sensor and actuator messages can assume arbitrary size, however, we assume that their size is fixed to a maximum size of $B$ bits.} % which can be achieved via padding.} 
%This models two practical aspects of our system. 
%First, WirelessHART devices cannot relay packets that have not been fully received.
%Second, the critical arrival at the first queue carries sensor value that must be processed by the controller to generate a command.

The end-to-end virtual delay, denoted by $W(k)$, is defined as
\begin{align}\label{eq:virtual-delay}
W(k) = \inf \left\lbrace u\geq 1 : \A(k) + x_2 \leq \D(k+u-1)\right\rbrace.
\end{align}
It quantifies the delay faced by the cumulative arrivals till frame $k-1$. 
%The virtual delay represents the time needed by the time critical arrival to traverse the network in presence of initial backlog.
%That is, for a single time critical arrival of $k_a$ packets at $t=0$
%\begin{align}
%\W(1) = \inf \left\lbrace w\geq0 : k_a + x_1 + x_2 \leq \D(1+w)\right\rbrace.

\subsection{\note{Lossy Wireless Network Model}}
\note{At the link layer, we consider an error-prone time-slotted system where multiple frequencies can be used for transmission.
Packet loss is caused by fading in the received signal, which can arise, for instance, from shadowing, mobility, or external interference. 
We assume that a frequency diversity mechanism is used in the network and sequential packet transmissions are characterized by uncorrelated channel fades.
Whenever critical messages are transmitted via unreliable wireless links, it is a common approach to deploy frequency diversity techniques, such as frequency hopping or frequency scheduling, to avoid sequential packet drops due to correlated channel fades. Thus, we restrict our analysis to the time domain.
}

We model the random service provided for a single packet transmission as a Bernoulli r.v. according to an average Packet Error Rate (PER) of the communication link.
That is, a packet is lost with probability $p_e$ and received with probability $1-p_e$. %$b_i$ % \sim \Be(1-p_e)$, where $p_e$ is the average PER of the channel.
%\begin{align}
%\Pr\left[b_i = b \right] &= \left\{
%\begin{array}{l l}
%1 & \ 1-p_e, \\
%0 & \ p_e.
%\end{array} \right.\label{eq:packet-erasure-process}
%\end{align}
\note{The PER is determined by the average Signal-to-Noise-and-Interference-Ratio (SINR) which in turn is determined by the combination of the propagation environment and the modulation and coding scheme used for transmission.} 

\note{Each frame comprises of $N$ time slots to be shared between the transmissions of packets from the two queues in the uplink and the downlink, cf. Fig.~\ref{fig:network-model}.}
In frame $k$, let $n_k^\text{1}$ and $n_k^\text{2}=N-n_k^\text{1}$ denote the slots used for transmitting the packets from the first queue and the second queue, respectively. 
Given this frame allocation, the service offered by the $i$-th link at frame $k$ is distributed as a Binomial r.v. given by
\begin{align}\label{eq:frame-service}
b_k^i(n_k^i) \sim \B\left(n^i_k ,1-p_e\right).
\end{align}
The cumulative service provided on the link at queue $i$ in $k$ consecutive frames is equal to a summation of Binomial random variables with parameters $1-p_e$, which is also a Binomial r.v. given by
%\begin{subequations}
\begin{IEEEeqnarray}{rCl}
	\IEEEeqnarraymulticol{3}{l}{
		\S^i(k) = \sum_{j=0}^{k-1} b^i_j(n^i_j) \sim \B\left(\sum_{j=0}^{k-1}n^i_j,1-p_e\right).
	}\label{eq:cumulative-service}
%\\\S^i(k,\vec{n}^i_k) &\sim& \B\left(\left\lVert\vec{n}^i_k \right\rVert_1,1-p_e\right).
\end{IEEEeqnarray}
%\begin{align}
%\S_i(t)&=\sum_{j=0}^{t-1} s_i \sim \B(t,p), \label{eq:cumulative-service}\\
%\P\{\S(t) = x\} &=\P\{\B(t,p) = x\} = {t\choose x} p^x (1-p)^{t-x}, \label{eq:binopdf}%\\ 
%%\P\{\B(t,p) \leq x\} &=\P\{\B(t,p) \leq x\} = \sum_{i=0}^{x} {t\choose i} p^x (1-p)^{t-i}. \label{eq:binocdf}
%\end{align}
%%\end{subequations}

\subsection{Problem Statement}
\note{We are interested in optimizing the dynamic service offered by the wireless links of sensor and controller to minimize the end-to-end delay of a time-critical arrival while it traverses the network.
In particular, in order to investigate scheduling policies that exploit initial network conditions, we study the impact of queue state information on the achievable performance of semi-static and dynamic resource allocations.

We define a scheduling policy $\pi$ as the allocation of time slots to both queues in every frame until the deadline, i.e. $\pi \triangleq \mathbf{n}^1 = \{n^\text{1}_0,n^\text{1}_1,\ldots,n^\text{1}_{w-1}\}$, equivalently $\pi \triangleq \mathbf{n}^2 = N-\mathbf{n}^1$.
%For a given fram $k$, the allocation of slots to the second queue is obtained by $n^\text{2}_k = N - n^\text{1}_k$.
Different scheduling algorithms are computed based on the queue state information $\mathbf{q}_k = (q^\text{1}_k,q^\text{2}_k)$, where $q^\text{1}_k$ and $q^\text{2}_k$ denote the lengths of first and second queues in frame $k$, respectively.

In the following, we consider scheduling policies that compute semi-static and dynamic resource allocations.
On the one hand, a semi-static scheduling policy, denoted by $\pi_S$, computes a schedule  based on the initial state $\mathbf{q}_0$.
Semi-static policies can be applied, for instance, to resource-constrained wireless networks such as WSN, where updating the network allocation over time is difficult due to the availability of a single radio interface and unreliable feedback channels. % that are costly and unreliable.
On the other hand, a dynamic scheduling policy, denoted by $\pi_D(\mathbf{q}_k)$, relies on the 
% This can be implemented in practical systems as 1) the controller knows the state of its own queue, and 2) knowing the initial state of the first queue, its evolution can be inferred by observing arrivals at the controller.
availability of the queue state $\mathbf{q}_k$ at a centralized network logic, which is used to determine the allocation of slots for the next frame, i.e. at $k$-th frame $n^1_{k}= \pi_D(\mathbf{q}_k)$.
An exemplary application of dynamic policies is in cellular networks, where reliable feedback channels can timely deliver new queue states to the network coordinator and resource allocations to the devices.} 

Given the end-to-end deadline $w$, we define the Delay Violation Probability (DVP) of a sequence of time-critical packets that arrived in frame $0$ as the probability that one or more packets of the sequence do not depart the second queue by the end of frame $w$. 
%\note{DVP directly relates to traditional real-time QoS metrics as it captures both delay and reliability requirements.}
For initial backlogs $x_1$, $x_2$ this is denoted by $\text{DVP}\left(w,y,x_1,x_2\right)$ and is given by
\begin{align}\label{def:DVP}
\text{DVP}\left(w,y,x_1,x_2\right) &\coloneqq \P\left\lbrace W(1) > w\right\rbrace.
\end{align}
%where $\P\left\lbrace\cdot\right\rbrace$ denotes the probability operator.
The above equivalence is obtained using Eq.~\eqref{eq:virtual-delay}, where the event $\{W(1) > w\}$ implies that the cumulative departures by the end of frame $w$ are smaller than the total number of packets in frame $0$. 
%This characterization is suitable for NCS where packets need to be delivered within the sampling period of the control loop.
Note that DVP could potentially be used as QoS in networked feedback systems; for example, given a deadline of $w$ frames, DVP represents the probability that a packet (carrying control command) in response to a packet generated by the sensor is delivered to the actuator within the deadline.

\note{We are interested in finding semi-static and dynamic scheduling policies that minimize the DVP of packets belonging to interactive applications.
The policies are obtained by formulating and solving the following optimization problems. Let $\Pi_S$ and $\Pi_D$ denote the sets of all possible semi-static and dynamic scheduling policies\footnote{\note{These sets are non-empty; an example policy allocates slots equally to both links in all frames.}}.
Given $y$ application packets arrived in frame $0$, an optimal semi-static policy $\pi_S$ is obtained solving
\begin{align}
	\underset{\pi_S \in \Pi_S}{\text{minimize}} \quad \text{DVP}\left(w,y,x_1,x_2\right).\label{eq:opt-offline}
\end{align}
A dynamic scheduling policy $\pi_D$ is obtained solving
\begin{align}
	\underset{\pi_D \in \Pi_D}{\text{minimize}} \quad \text{DVP}\left(w,y,x_1,x_2\right).\label{eq:opt-online}
\end{align}
In Eq.~\eqref{eq:opt-offline} and Eq.~\eqref{eq:opt-online}, $\Pi_S, \Pi_D$ denote the sets of all possible semi-static and dynamic scheduling policies, and are non-empty as each resulting slot allocation is valid.

In the sequel, we will be using the following definitions.
Given a set of events $E_1, E_2, \dots$ the \emph{union bound} is given by
\begin{align*}
    \P\left\lbrace\bigcup_{i}^{} E_i\right\rbrace \leq \P\left\lbrace \sum_{i}^{} E_i \right\rbrace.
\end{align*}
Furthermore, given a random variable $X$, for every $t>0$, the \emph{Chernoff bound} is given by
\begin{align*}
    \P\left\lbrace X \geq x \right\rbrace \leq \frac{\E[e^{tX}]}{e^{tx}},
\end{align*}
where $\E[\cdot]$ denotes the expectation operator. 
%\sz{add here Union and Chernoff bounds, but no eq. number, then, we use Chernoff Bound in Eq.()...}

}
\section{Derivation of Delay Violation Probability}\label{sec:dvp-heuristic}
We characterize DVP using Stochastic Network Calculus (SNC)~\cite{Jiang:2006:BSN:1151659.1159929}. 
From Eq.~\eqref{eq:virtual-delay}, DVP can be obtained in terms of the virtual delay of the network
\begin{align}
\text{DVP}\left(w,y,x_1,x_2\right) &= \P\left\lbrace W(1) > w\right\rbrace \nonumber\\
&\,= \P\left\lbrace\D(w) < y + x_1 + x_2\right\rbrace. \label{eq:dvp-def}
\end{align}

Applying the input-output relation for a queue with a dynamic server
\begin{equation}\label{eq:dynamic-server}
\D(k) \geq \min_{0 \leq u \leq k} \left[\A(u)+\S(k-u)\right],
\end{equation} 
we can derive the exact expression of DVP.
\begin{proposition}\label{prop1}
The delay violation probability (DVP) of a time critical arrival of $y$ packets at $k=0$, given initial queue backlogs $x_1,x_2$ is 
\begin{IEEEeqnarray}{rCl}
	\IEEEeqnarraymulticol{3}{l}{
		\text{\emph{DVP}}(w,y,x_1,x_2) = 
	}\nonumber\\* \quad
	&& \P\Big\{ \!\left\{\S^2(w) \! < \!y + x_1 + x_2\right\} \cup \left\{\S^2(w-1) \! < \! y + x_1\right\} \cup \nonumber\\
	&&\bigcup_{u=2}^{w} \left\{\S^2(w-u) + \S^1(u-1) < y + x_1\right\}\Big\} \label{eq:exact-dvp2q}.
\end{IEEEeqnarray}
\end{proposition}
\begin{proof}
The proof can be found in Appendix~\ref{app:dvp-derivation}
\end{proof}

From Eq.~\eqref{eq:exact-dvp2q}, we observe that the calculation of DVP is highly non-trivial.
The DVP computation requires the knowledge of future, i.e. both the allocations $n^1_k$ and the resulting queue states, in order to calculate the cumulative services.
Thus, it is impossible to use DVP to obtain a scheduling policy which causally allocates the time slots \note{at each frame within the deadline} using only the past information.
Furthermore, computing the exact value of DVP is not tractable as it requires the calculation of the probability of union of $w$ events that are not mutually disjoint.

\note{In this work, we address this issue by deriving upper bounds for DVP which are then used to design semi-static and dynamic schedulers. %, and applied to different wireless networks. %in order to derive scheduling policies based on DVP.
In Sec.~\ref{sec:offline-scheduling}, we investigate \emph{semi-static} scheduling policies %for infrastructure-less wireless networks 
based on two upper bounds of Eq.~\eqref{eq:exact-dvp2q} to determine the allocation of slots until the deadline solely relying on the initial queue states.
Then, in Sec.~\ref{sec:online-scheduling}, a \emph{dynamic} scheduling policy is derived based on another upper bound for DVP that reallocates the network resources according to the changes in the queue states.}

\note{\section{Semi-static Scheduling Policies}\label{sec:offline-scheduling}
%The derivation of exact DVP values is not tractable as it requires the calculation of the probability of union of $w$ events that are not mutually disjoint.
In this section, we derive an upper bound for DVP, referred to as DVPUB, and formulate an upper bound minimization problem, which is then used to compute the proposed semi-static policies. 

Using the union bound for DVP in~\eqref{eq:exact-dvp2q}, we obtain DVPUB, given by
\begin{align}
		&\text{DVPUB}(w,y,x_1,x_2) = 
	\P\left\{\S^2(1+w) \! < \! y + x_1 + x_2\right\} + \nonumber\\
	& \sum_{u=1}^{1+w} \P\left\{\S^2(1+w-u) + \S^1(u-1) < y + x_1\right\} \nonumber\\
	& = \P\left\{\S^2(1+w) \! \leq \! y + x_1 + x_2 - 1\right\} + \nonumber\\
	%&& \P\left\{\S_2(w) \! \leq \! y + x_1 - 1\right\} + \nonumber\\
	& \sum_{u=1}^{1+w} \P\left\{\S^2(1+w-u) + \S^1(u-1) %\nonumber\\ &&
	 \leq\! y + x_1 -1\right\}.\label{eq:dvpub}
\end{align}

%DVPUB is the sum of CDF of the cumulative service of the second server $\S_2$ and the CDF of the sum of the cumulative services of servers $\S_1,\S_2$.
Applying Eq.~\eqref{eq:cumulative-service} to Eq.~\eqref{eq:dvpub}, the DVPUB resulting from the allocation of $\mathbf{n}^1 = \{n^1_0, \dots, n^1_w-1\}$ slots at the second transmitter and $\mathbf{n}^2=N-\mathbf{n}^1$ slots at the first one is given by
\begin{align}
		&\text{DVPUB}(w,y,x_1,x_2,\mathbf{n}^2) =%= \P\{\W(1) > w\}
	%	&=&\P\left\{\S_2(1+w) \! < \! y + x_1 + x_2\right\} + \nonumber\\
	%	&& \sum_{u=1}^{1+w} \P\left\{\S_2(1+w-u) + \S_1(u) < y + x_1\right\}\nonumber\\
	\sum_{x=0}^{y+x_1+x_2-1}\P\left\{\sum_{i=0}^{w} b_i^2\left(n_i^2\right) = x\right\}\nonumber\\
	%&&+\sum_{y=0}^{y+x_1-1}\P\left\{\B\left(\sum_{i=0}^{w-1} n_i^2,1-p_e\right) = y\right\}\nonumber\\
	&+\sum_{u=1}^{1+w} \sum_{z=0}^{y+x_1-1} \P\left\{\sum_{j=0}^{w-u} b_j^2\left(n_j^2\right)+\sum_{k=0}^{u-2}b_k^1 \left( n_k^1\right) = z \right\} \nonumber\\
	&\symtext{(a)}{=}\sum_{x=0}^{y+x_1+x_2-1} \left(\frac{1}{p_e}-1\right)^x {\sum_{i=0}^{w}n_i^2\choose x} p_e^{\sum_{i=0}^{w}n_i^2}\notag\\
	%&&+ \sum_{y=0}^{y+x_1-1} \left(\frac{1}{p_e}-1\right)^x {\sum_{i=0}^{w-1}n_i^2\choose y} p_e^{\sum_{i=0}^{w-1}n_i^2}\notag\\
	&+\sum_{u=1}^{1+w} \sum_{z=0}^{y+x_1-1} \left(\frac{1}{p_e}-1\right)^z  p_e^{\sum_{j=0}^{w-u}n_j^2-\sum_{k=0}^{u-2}n_k^2+(u-1)N} \cdot \nonumber\\
	&{\sum_{j=0}^{w-u}n_i^2-\sum_{k=0}^{u-2}n_j^2+(u-1)N\choose z}.\label{eq:dvpub-w-server}
\end{align}
In step (a), we used the fact that $S^i(k)$ is distributed as a Binomial r.v. as shown in Eq.~\eqref{eq:cumulative-service}. 

We are interested in minimizing the DVPUB to find semi-static scheduling policies.
However, this is highly non trivial for different reasons. 
Minimizing DVPUB is a combinatorial problem and finding a heuristic solution by relaxing the domain of $\mathbf{n}^2$ is challenging as DVPUB consists of the sum of several binomial coefficients.
Thus, it is highly non trivial to study its convexity.
Therefore, following a similar approach as in~\cite{Champati2020TransientRouting} a looser convex bound, referred to as Wireless Transient Bound (WTB), is obtained
by applying the \emph{Chernoff bound} to Eq.~\eqref{eq:dvpub}.
\begin{align}
		&\text{WTB}(w,y,x_1,x_2) =
	\min_{s>0} \Big\{\E\left[e^{-s \, \S^2(1+w)}\right] e^{s\,(y + x_1 +x_2 -1)} + \nonumber\\
	%&& \E\left[e^{-s \, \S_2(w)}\right] e^{s\,(y + x_1 -1)} + \nonumber\\
	&\quad\sum_{u=1}^{1+w} \E\left[e^{-s \left[ \S^1(u-1)+\S^2(1+w-u)\right]}\right] e^{s\,(y + x_1 -1)}\Big\}.\label{eq:wtb}
\end{align}

The calculation of WTB for a wireless transmitter is obtained by computing the Mellin transform of the cumulative service of Eq.~\eqref{eq:cumulative-service}, given by
\begin{align}
	\mathbb{E}\left[e^{-s \, \S(m,\mathbf{n})}\right] &= \mathbb{E}\left[e^{-s \, \sum_{i=0}^{-1+\sum_{j=0}^{m-1} n_j} b_i}\right] \nonumber\\
	&=\mathbb{E}\left[\left(e^{-s \, b_i}\right)^{\sum_{i=0}^{m-1} n_i}\right]\nonumber\\
	&= \left[(1-p_e)e^{-s}+p_e\right]^{\sum_{i=0}^{m-1} n_i},\label{eq:mellin-dynamic-server} 
\end{align}
where $b_i$ are i.i.d. Bernoulli random variables for transmission outcomes.
Finally, combining Eq.~\eqref{eq:wtb} and~\eqref{eq:mellin-dynamic-server}, we obtain
\begin{align}
		&\text{WTB}(w,y,x_1,x_2,\mathbf{n}^2) = \nonumber\\
	%&=&\min_{s>0} \E\left[e^{-s \, \S_2(1+w)}\right] e^{s\,(y + x_1 +x_2 -1)} + \nonumber\\
	%&& \sum_{u=1}^{1+w} \E\left[e^{-s \left[ \S_1(u)+\S_2(1+w-u)\right]}\right] e^{s\,(y + x_1 -1)} \nonumber\\
	%&=& \min_{s>0} \mathbb{E}\left[e^{-s \, \sum_{i=0}^{-1+\sum_{j=0}^{w} n_j^2} b_i}\right] e^{s\,(y + x_1 +x_2 -1)} +\nonumber\\
	%&& \sum_{u=1}^{1+w} \E\left[e^{-s  \left[\sum_{i=0}^{-1+\sum_{j=0}^{u-1-1} N-n_j^2+\sum_{j=0}^{w-u} n_j^2} b_i\right]}\right] e^{s\,(y + x_1 -1)} \nonumber\\
	&\min_{s>0} \Big\{ \left[(1-p_e) e^{-s}+p_e \right]^{\sum_{i=0}^{w} n_i^2} e^{s\,(y + x_1 +x_2 -1)} + \nonumber\\
	&\sum_{u=1}^{1+w} \!\left[(1\!-\!p_e) e^{-s}\!\!+p_e\right]^{(u-1)N-\sum_{j=0}^{u-2} n_j^2+\sum_{i=0}^{w-u} n_i^2} \!e^{s\,(y + x_1 -1)}\Big\}.\label{eq:wtb-w-server}
\end{align}

Given the initial queue backlogs $\mathbf{q}_0 = \left\lbrace x_1,x_2\right\rbrace$, we aim to solve the upper bound minimization problem below
% the semi-static scheduling policy, denoted by $\pi^\star_1\left(\mathbf{q}_0\right)$, that minimizes Eq. \eqref{eq:wtb-w-server}, is given by
\begin{equation}
\argmin_{\mathbf{n}^2 \in \left\{1,\dots,N-1\right\}^w}\text{WTB}(w,y,x_1,x_2,\mathbf{n}^2) .\label{eq:wtb-opt}
\end{equation}
To solve the integer-programming problem \eqref{eq:wtb-opt}, one may employ exact algorithms (such as \textit{branch-and-bound}), which however has run-time that scales exponentially with $N$ and $w$. Instead, we relax the integer constraints, show that the relaxed problem is convex, and use different methods to round the  continuous values of the solution and obtain multiple heuristics for the minimization of DVPUB.

Let $\widetilde{\pi}^\star_S$ denote the optimal solution for the relaxed problem of~\eqref{eq:wtb-opt} which is given by
%As shown by Theorem~\ref{th1}, a solution to the problem of Eq.~\eqref{eq:wtb-opt} can be efficiently derived when the domain of the optimization variable $\mathbf{n}^2$ is continuous.
%This can be achieved by solving
\begin{equation}
\widetilde{\pi}^\star_S \, = \argmin_{\mathbf{n}^2 \in \left[1,N-1\right]^w}\text{WTB}(w,y,x_1,x_2,\mathbf{n}^2).\label{eq:wtb-opt-relaxed}
\end{equation}
%then, by converting $\widetilde{\pi}^\star_S$ to the integer domain to find a feasible solution to Eq.~\eqref{eq:opt-static}, thus obtaining the scheduling policy $\hat{\pi}^\star_S$.
\begin{theorem}\label{th1}
	The optimization problem in~\eqref{eq:wtb-opt-relaxed} is convex.
\end{theorem}
\begin{proof}
The proof is given in Appendix~\ref{app:wtb2q-convexity}.
\end{proof}
Thanks to Theorem~\ref{th1}, well-known convex optimization algorithms, such as the subgradient or interior-point can be used, which provide scalability for an increasing number of slots $N$ and frames within the deadline $w$.
In this work,~\eqref{eq:wtb-opt-relaxed} is solved using the nonlinear programming solver \emph{fmincon} available in Matlab\texttrademark~employing the sequential quadratic programming (SQP) algorithm.
Once $\widetilde{\pi}^\star_S$ is found, a conversion to the integer domain is needed in order to determine a feasible solution, which we refer by $\hat{\pi}^\star_S$.
Although the optimal selection of an integer solution would require the exploration of the entire problem's domain, heuristic methods can be applied to find a solution in the neighbourhood of $\widetilde{\pi}^\star_S$.
To this end, we investigate different neighbour search methods in order to achieve near-optimal performances.

The simplest way to derive $\hat{\pi}^*_{S}$ is to round each frame allocation of $\widetilde{\pi}^\star_S$ to its closest integer value.
We refer to this method as WTB-R.
Alternatively, a heuristic policy can be found as follows.
For each frame allocation $n^1_i \in \widetilde{\pi}^\star_S$, two integer values are derived applying the floor and ceiling functions.
A search space is constructed by computing all the combinations of the integer values for each frame until the deadline which leads to a total of $2^w$ combinations.
As final step, Eq.~\eqref{eq:dvpub-w-server} and \eqref{eq:wtb-w-server} are used to evaluate each combination and identify the best one.
The semi-static policies corresponding to the evaluation of Eq.~\eqref{eq:dvpub-w-server} and \eqref{eq:wtb-w-server} over this search space are referred to as WTB-D and WTB-W, respectively.

The performances of the different heuristics have been evaluated via extensive simulations over a broad range of values for each system parameter, i.e. for $x_1, x_2$~$\in$ ~$\left\{0,1,2\right\}$, $w$~$\in$~$\left\{2,3,4,5,6\right\}$, $N$~$\in$~$\left\{2,3,4,5\right\}$, and $p_e$~$\in$~$\left\{0.2,0.33,0.4,0.5\right\}$.
For performance comparison, two additional schedulers have been evaluated, eDVPUB and eWTB, which exhaustively explore the problem's domain to find the policies that respectively achieve the minimum values of DVPUB, cf. Eq.~\eqref{eq:dvpub-w-server}, and WTB, cf. Eq.~\eqref{eq:wtb-w-server}.
Furthermore, the performance of each scheduler is compared with the performance of the policy that achieves the minimum DVP.
This optimal policy is found by exploring the DVP of all possible policies via simulations and then, for each policy, simulations are used to compute its DVP.
%Thus, the policies obtained by eDVPUB and eWTB are computed evaluating $N^w$ times .

Fig.~\ref{fig:1-boxplot-offline} shows the performance of the proposed semi-static schedulers by computing, for each scheduler and system configuration, the percentage of feasible policies that achieve higher DVP.
% This calculat
% with respect to the optimal policy that achieves the minimum DVP.
% For each parameter, the performance of each policy is expressed using the ranking the performance
As expected, the exhaustive search methods eWTB and eDVPUB achieve the best results, with eDVPUB finding, in the large majority of cases, the top $10\%$ semi-static policies.
The performance gap between eWTB and eDVPUB is introduced by the Chernoff bound in Eq.~\eqref{eq:mellin-dynamic-server}.
These methods, however, do not represent a feasible way of computing policies as they require the exhaustive exploration of the entire problem domain.
Differently, the heuristic methods WTB-R/W/D can be used to efficiently find semi-static policies for arbitrary parameters.
The low complex WTB-R achieves the worst performance.
WTB-W and WTB-D achieve similar performances and show a small performance penalty with respect to the exhaustive methods, with WTB-W achieving better performances.
This can be explained by the fact that, in WTB-D, Eq.~\eqref{eq:wtb-w-server} is used to solve the relaxed problem, while Eq.~\eqref{eq:dvpub-w-server} to find the integer solution.
For this reason, only WTB-W is presented in the numerical section.
%This discrepancy results in reduced performance.
}

\begin{figure}[!t]
	\centering
	\includegraphics{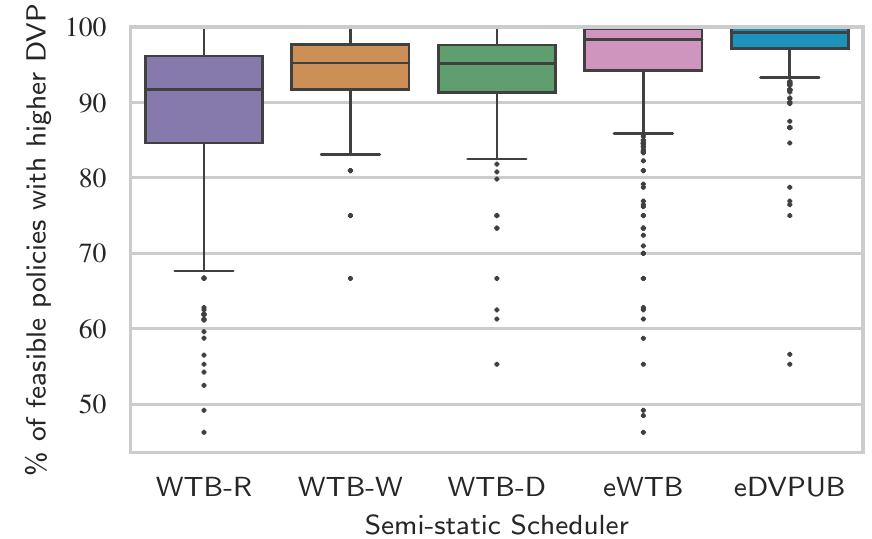}
	\caption{\note{Performance of the proposed semi-static schedulers for different system parameters. For each scheduler, the percentage of existing policies that achieve higher DVP is evaluated.}}
	\label{fig:1-boxplot-offline}
\end{figure}
\section{Dynamic Scheduling Policy}\label{sec:online-scheduling}
In order to compute the DVP of time-critical packets, cf. Eq.~\eqref{eq:exact-dvp2q}, the knowledge of future allocation and queue states is needed.
Thus, it is impossible to use DVP to obtain a dynamic scheduling policy which causally allocates the time slots using only the past information. 
To address this problem, we obtain an upper bound for DVP using Markov's inequality: \begin{comment}
\begin{align}
		\text{DVP}(w,y,x_1,x_2) %\P\{\D(1+w) \leq y + x_1 + x_2 - 1\}
	&=  \P\{\D(w) < y + x_1 + x_2\} \nonumber\\
	&=  \P\{\D(w) \leq y + x_1 + x_2 - 1\} \nonumber\\
	&=\P\left\{y + x_1 + x_2 -\D(w) \geq 1\right\} \nonumber\\
	&\leq y + x_1 + x_2 - \E[\D(w)]. \label{eq:dvp-heuristic-2}
\end{align}
\end{comment}
\note{
\begin{align}
		\text{DVP}(w,y,x_1,x_2) %\P\{\D(1+w) \leq y + x_1 + x_2 - 1\}
	&=  \P\{\D(w) < y + x_1 + x_2\} \nonumber\\
	&=  \P\{\D(w) \leq y + x_1 + x_2 - 1\} \nonumber\\
	&=  \P\{1/\D(w) \geq 1/(y + x_1 + x_2 - 1)\} \nonumber\\
	&\leq (y + x_1 + x_2)\E[1/\D(w)]. \label{eq:dvp-heuristic-2}
\end{align}
%From Eq.~\eqref{eq:dvp-heuristic-2} we infer that maximizing the expected cumulative departures (throughput) of the network minimizes the upper bound of the DVP and thus potentially reduces DVP. 
From Eq.~\eqref{eq:dvp-heuristic-2} we infer that minimizing the expectation of the inverse of the cumulative departures (throughput) of the network minimizes the upper bound of the DVP and thus potentially reduces DVP. 
Using this insight, in the following, we compute a heuristic schedule by solving the expected throughput maximization problem stated below:
%For this reason, the DVP of a single time critical arrival can be improved solving a throughput maximisation optimisation problem

\begin{align}\label{eq:throughput-max-problem}
\underset{\pi_D \in \Pi_D}{\text{maximize}} \quad \E[\D(w)]= \sum_{k=0}^{w-1} \E\left[d^i_k\right] .
\end{align}
}

In order to solve the optimization problem in Eq.~\eqref{eq:throughput-max-problem}, we formulate a discrete-time, finite-horizon MDP.
% and evaluate its QoS performance in terms of DVP.
We use $\mathbf{q}_k$ to denote the state of the system and $n^\text{1}_k$ to denote the action in frame $k$. 
The maximum number of slots in a frame is $N$ and therefore $n^\text{1}_k \in \{0,1,\ldots,N\}$. 
%Given $n^\text{1}_k$ the number of slots allocated to the second queue is uniquely determined by $n^\text{2}_k = N-n^\text{1}_k$. 
%In frame $k$ we denote the service received by first and second queues by $s_k^\text{1}$ and $s_k^\text{2}$, respectively. 
Given $n^\text{1}_k$, from (6) we have
\begin{align*}
\mathbb{P}\{s_k^\text{1} = r\} &= \binom{n^\text{1}_k}{r} (1-p_e)^r p_e^{n^\text{1}_k-r}, \\ 
\mathbb{P}\{s_k^\text{2} = r\} &= \binom{N-n^\text{1}_k}{r} (1-p_e)^r p_e^{N-n^\text{1}_k-r}.
\end{align*}
The queues evolve as below:
\begin{align}
q^\text{1}_{k+1} &= \max(q^\text{1}_k -  s_k^\text{1},0), \label{eq:q1}\\
q^\text{2}_{k+1} &= \max(q^\text{2}_k -  s_k^\text{2},0) + \min(q^\text{1}_k,s_k^\text{1}). \label{eq:q2}
\end{align}
Note that the number of departures from the first queue in frame $k$ equals $\min(q^\text{1}_k,s_k^\text{1})$, which are added to the second queue to be served in frame $k+1$.
%\footnote{Samuele: This is unlike in our transient analysis where we assume that the departures from first queue will be added to the second  queue and get served in the same slot. I suggest we go with the above queue update model which is more practical. This means that we have to redo the DVP and WTB bounds.}

In the following, we formulate the transition probabilities for the states. Note that the initial backlogs in the queues are $(y+x_1,x_2)$, where $y$ is due to the message of interest. We have $q^\text{1}_0 = y + x_1$ and $q^\text{2}_0 = x_2 $.  We now analyse the set of possible states in our system. In any frame $k$, a feasible state $(q^\text{1}_k,q^\text{2}_k)$ should satisfy the following conditions:
\begin{align}
q^\text{1}_k &\leq q^\text{1}_{k-1}, \label{eq:cond1}\\
q^\text{1}_k + q^\text{2}_k &\leq q^\text{1}_{k-1} + q^\text{2}_{k-1}.\label{eq:cond2}
\end{align}
Conditions Eq.~\eqref{eq:cond1} and Eq.~\eqref{eq:cond2} follow from the fact that we ignore arrivals after the message of interest and in every frame each queue will receive certain service. Note that while the length of the first queue can only decrease as the packets are served, the length of the second queue may increase up to $y+x_1+x_2$ as departures from first queue are added to the second queue. Therefore, for every state $\mathbf{q}_k$  in the state space, say $\mathcal{Q}$, $q^\text{1}_k \in \{0,1,\ldots,y+x_1\}$ and $q^\text{2}_k \in \{0,1,\ldots,y+x_1+x_2\}$. \note{This implies that $\mathcal{Q}$ can contain at most $(y+x_1+1)(y+x_1+x_2+1)$ possible states.}

Consider that in frame $k$,  $q^\text{1}_k = l^1$ and $q^\text{2}_k = l^2$. We would like to present the transition probabilities to the states  $q^\text{1}_{k+1} = l^1_+$ and $q^\text{2}_{k+1} = l^2_+$. We have the following cases.

\textbf{Case 1:}  $l^1_+ > l^1$ or  $l^1_+ + l^2_+ > l^1 + l^2$. From Eq.~\eqref{eq:cond1} and Eq.~\eqref{eq:cond2}, we infer that
\begin{align*}
\mathbb{P}\{q^\text{1}_{k+1} = l^1_+,q^\text{2}_{k+1} = l^2_+ | q^\text{1}_{k} = l^1,q^\text{2}_{k} = l^2\} = 0.
\end{align*}

\textbf{Case 2:}  $0 < l^1_+ \leq l^1$, $0 < l^2_+$, and  $l^1_+ + l^2_+ \leq l^1 + l^2$. In this case $ s_k^\text{1} < q^\text{1}_k = l^1$ and $ s_k^\text{2} < q^\text{2}_k = l^2$. From Eq.~\eqref{eq:q1} we have 
\begin{align*}
q^\text{1}_{k+1} = q^\text{1}_k -  s_k^\text{1} \quad
\Rightarrow  \quad s_k^\text{1}  = l^1 - l^1_+.
\end{align*}
The number of packets served from the second queue are computed from Eq.~\eqref{eq:q2}. 
\begin{align*}
&q^\text{2}_{k+1} = q^\text{2}_k -  s_k^\text{2}  + s_k^\text{1} \quad
\Rightarrow  s_k^\text{2}  = l^2 - l^2_+ + l^1 - l^1_+.
\end{align*}
Therefore,
\begin{align*}
\mathbb{P} \{q^\text{1}_{k+1} &= l^1_+,q^\text{2}_{k+1} = l^2_+ | q^\text{1}_{k} = l^1,q^\text{2}_{k} = l^2\}\\ &= \mathbb{P}\{s_k^\text{1}  = l^1 - l^1_+, s_k^\text{2}  = l^2 - l^2_+ + l^1 - l^1_+\}.
\end{align*}

\textbf{Case 3:} $l^1_+ = 0$, $0 < l^2_+$, and  $l^2_+ \leq l^1 + l^2$. In this case all $l^1$ packets from the first queue are served. This implies $ s_k^\text{1} \geq q^\text{1} = l^1$. Using similar analysis as above, we obtain
\begin{align*}
\mathbb{P}\{q^\text{1}_{k+1} =0,q^\text{2}_{k+1} = l^2_+ | q^\text{1}_{k} = l^1,q^\text{2}_{k} = l^2\}\\ = \mathbb{P}\{s_k^\text{1}  \geq l^1 , s_k^\text{2}  = l^2 - l^2_+ + l^1\}.
\end{align*}

\textbf{Case 4:}  $l^1_+ = l^1$, $l^2_+ = 0$. In this case we have $s_k^\text{1}  = 0$, and all $l^2$ packets from the second queue are served, i.e. $ s_k^\text{2} \geq q^\text{2} = l^2$.  From Eq.~\eqref{eq:q2}, we have
\begin{align*}
\mathbb{P}\{q^\text{1}_{k+1} &=l^1,q^\text{2}_{k+1} = 0 | q^\text{1}_{k} = l^1,q^\text{2}_{k} = l^2\} \\&= \mathbb{P}\{s_k^\text{1}  = 0 , s_k^\text{2}  \geq l^2\}.
\end{align*}

Note that the case $0 \leq  l^1_+ < l^1$ and $l^2_+ = 0$ cannot happen as $l^1 - l^1_+$ packets will be added to the second queue in the current slot. All the above cases are written assuming that $l^1 > 0$ and $l^2 > 0$. If either of them is zero, then the transition probability only involves the probability for service received by the non-empty queue.

Given the initial state  $\mathbf{q}_0 = (y+x_1,x_2)$, we are interested in finding a scheduling policy $\pi^\star_D$ that solves the maximization problem of Eq.~\eqref{eq:throughput-max-problem}.
%\begin{align}\label{eq:mdp-objective}
%\underset{\pi \in \Pi}{\text{maximize}} \quad \sum_{k=0}^{w} \mathbb{E}\left[d^i_k|\,\pi(\mathbf{q}_k)\right], \quad \mathbf{q}_k\in\mathcal{Q} %= \sum_{k=0}^{w} \mathbb{E}\left[\min\left(q^\text{2}_{k}, s_k^\text{2}\right)|\,\pi(\mathbf{q}_k)\right].
%\end{align}
For this, we define the reward $r_k$ of a policy $\pi_D$ for a given state $\mathbf{q}_k$ as the expected number of departures from the system, i.e. the expected number of packets that are served at the second queue under this policy, and is given by
\begin{align}\label{eq:mdp-reward}
r_k\left(\mathbf{q}_k,\pi_D(\mathbf{q}_k)\right) &= \E\left[d^i_k|\,\pi_D(\mathbf{q}_k)\right] \nonumber\\
&= \E\left[\min\left(q^\text{2}_{k}, s_k^\text{2}\right)|\,\pi_D(\mathbf{q}_k)\right].
\end{align}

The total reward, obtained evaluating Eq.~\eqref{eq:mdp-reward} over a horizon of $w$ frames, is equal to Eq.~\eqref{eq:throughput-max-problem}.
Therefore, the objective of the MDP is equal to the objective of Eq.~\eqref{eq:throughput-max-problem}.

Value iteration algorithms solve the MDP optimization recursively computing a value function $J$ based on the Bellman's equation~\cite{Puterman:1994:MDP:528623}. 
The optimal value function $J(\mathbf{q}_k)$ given a state $\mathbf{q}_k$ is
\begin{align}
J_k(\mathbf{q}_k)=&\,\underset{\pi_D \in \Pi_D}{\text{maximize}}\,\,\,\, r_k(\mathbf{q}_k,\pi_D(\mathbf{q}_k))\,+\nonumber\\ &\sum_{\mathbf{q}_{k+1}\in\mathcal{Q}_{k+1}}\P\{\mathbf{q}_{k+1}|\mathbf{q}_k,\pi_D(\mathbf{q}_k)\}J_{k+1}(\mathbf{q}_{k+1}), \label{eq:J-transition-probability}
\end{align}
where $\mathbb{P}\{\mathbf{q}_{k+1}|\mathbf{q}_k\pi_D(\mathbf{q}_k)\}$
is the transition probability from state $\mathbf{q}_k$  to state $\mathbf{q}_{k+1}$
in one time step using $\pi_D(\mathbf{q}_k)$, $\mathcal{Q}_{k+1}$ denotes the set of all states reachable from $\mathbf{q}_k$ with one time step transition.
By the construct of the MDP, it is easy to see that $\pi_D^\star$ is optimal for Eq.~\eqref{eq:throughput-max-problem} which is stated in the following theorem.
%That is, Eq.\eqref{eq:J-transition-probability} for $k=w-1,\dots,0$, and $J_w(\mathbf{q}_i)=c_w(\mathbf{q}_i,\pi(\mathbf{q}_i))$, $\forall \mathbf{q}_i \in \mathcal{Q}$.
\begin{theorem}\label{th2}
	$\pi_D^\star$ is throughput optimal, i.e.
	\begin{align*}
	\pi_D^\star = \underset{\pi_D \in \Pi_D}{\argmax} \quad \E\left[\D(w)\right].    
	\end{align*}
\end{theorem}

%\subsection{Online Policy Computation}\label{subsec:online-computation}
\note{For a finite number of states and actions, the optimal policy $\pi_D^\star$ for the MDP can be found by computing the optimal value function in Eq.~\eqref{eq:J-transition-probability} by backward recursion (cf.~\cite{Puterman:1994:MDP:528623}). 
The calculation of the optimal dynamic scheduling policy satisfying Theorem~\ref{th2} can be performed using the finite-horizon value iteration algorithm.
%Value iteration can efficiently compute the optimal policy, particularly for finite-horizon problems.
For each epoch $k$ until the deadline $w$ and state $\mathbf{q}_{k}$, the value function in Eq.~\eqref{eq:J-transition-probability} is computed for all possible actions. 
Therefore, computing the optimal dynamic policy requires $\mathcal{O}\left(\left(N+1\right)|\mathcal{Q}|w\right)$ operations. We note that optimal actions for all the states in $\mathcal{Q}$ can be computed a priori and stored in a table, and for a queue state observed in a frame the corresponding optimal action can be retrieved from the table.}
\section{Performance Evaluation of Semi-static and Dynamic Scheduling Policies}\label{sec:results}
In this section, the performance of semi-static and dynamic scheduling policies are evaluated numerically.
The evaluations leverage C, Matlab\texttrademark~and python to compute the performance of, respectively, semi-static and dynamic policies and comparison baselines.
We evaluate the schemes for variations in several main system parameters: (1) Different initial queue backlogs $x_1, x_2$; (2) Application deadlines $w$; (3) Number of slots per frame $N$; and (4) Average service PERs $p_e$.
In Sec.~\ref{subsec:offline-eval}, we present first a performance comparison of purely semi-static scheduling policies. 
In Sec.~\ref{subsec:online-eval}, we then move to purely dynamic policies.
Finally, in Sec.~\ref{subsec:off-on-eval}, we present results on the performance gap between the proposed semi-static and dynamic scheduling policies. 
%Finally, Sec.~\ref{subsec:eval-discussion} provides an overview discussion on the proposed methods with respect to the queue state information and their application in real networks.

\begin{figure*}[!t]
	\centering
	\begin{minipage}{\columnwidth}
		\includegraphics{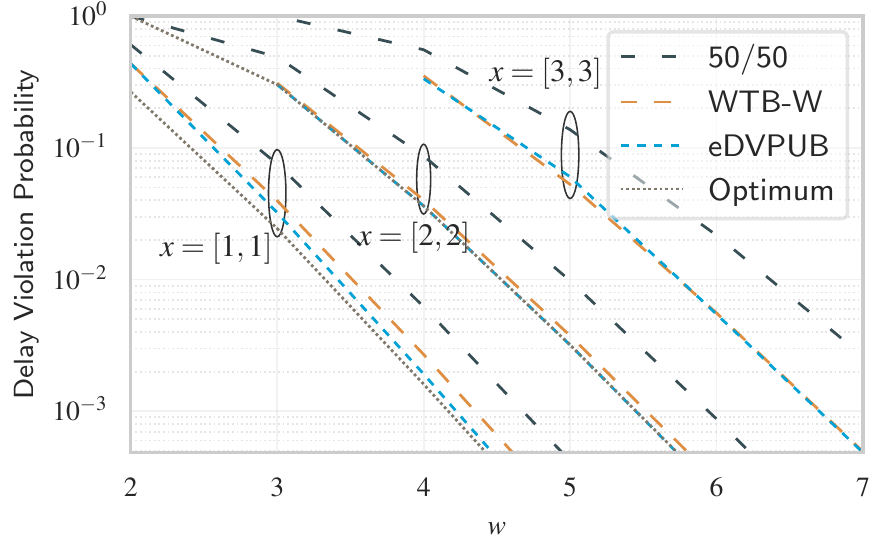}
		\caption{\note{DVP achieved by semi-static schedulers for different deadlines $w$, increasing backlogs $x_1, x_2$, $N=4$, $p_e=0.2$.}}
		\label{fig:fig_lineplot_offline-high_xsymm_x1s_x2s_pes2_Ns4_ws_algs_X_w__Y_simdvp__Z_x1,1_2,2_3,3_}
	\end{minipage}\hfill
	\begin{minipage}{\columnwidth}
		\includegraphics{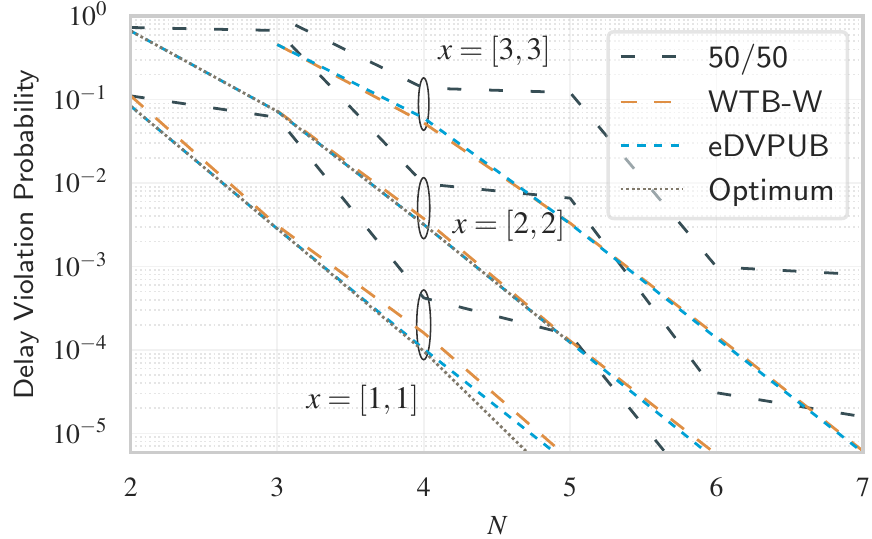}
		\caption{\note{DVP achieved by semi-static schedulers for different frame sizes $N$, increasing backlogs $x_1, x_2$, $w=5$, $p_e=0.2$.}}
		\label{fig:fig_lineplot_offline-high_xsymm_x1s_x2s_pes2_Ns_ws5_algs_X_N__Y_simdvp__Z_x1,1_2,2_3,3_}
	\end{minipage}
\end{figure*}
\begin{figure*}[!t]
	\centering
	\begin{minipage}{\columnwidth}
		\includegraphics{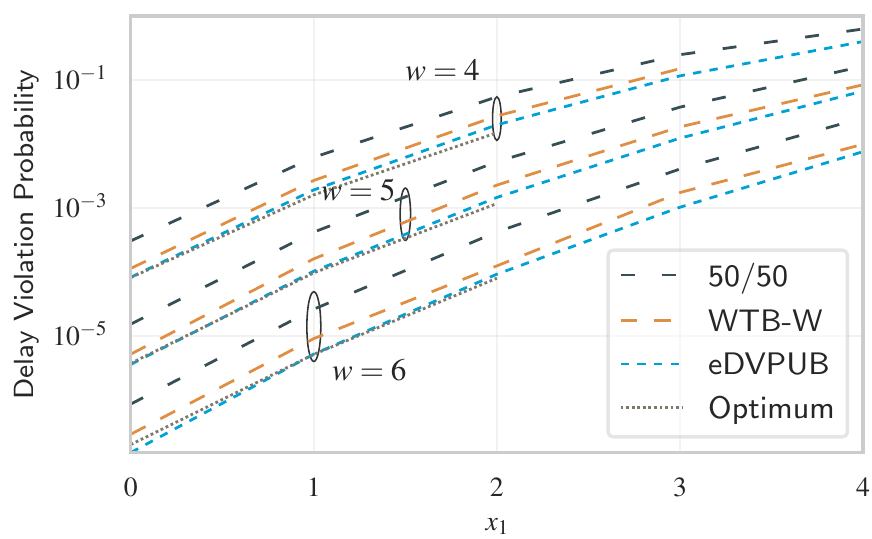}
		\caption{\note{DVP achieved by semi-static schedulers for different backlogs $x_1$ and deadlines $w$, $x_2=1$, $N=4$, $p_e=0.2$.}}
		\label{fig:fig_lineplot_offline-high_x1inc_x1s_x2s_pes2_Ns4_ws_algs_X_x1__Y_simdvp__Z_w4,5,6_}
	\end{minipage}\hfill
	\begin{minipage}{\columnwidth}
		\includegraphics{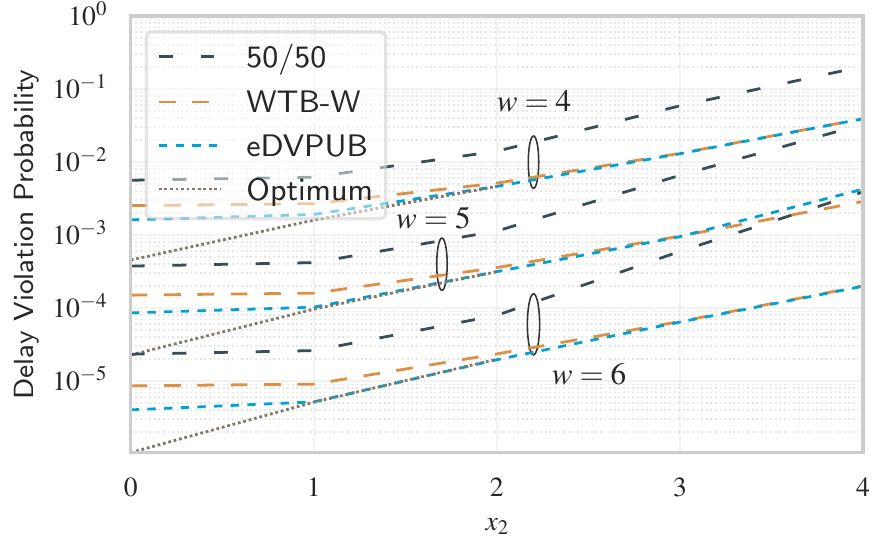}
		\caption{\note{DVP achieved by semi-static schedulers for different backlogs $x_2$ and deadlines $w$, $x_1=1$, $N=4$, $p_e=0.2$.}}
		\label{fig:fig_lineplot_offline-high_x2inc_x1s_x2s_pes2_Ns4_ws_algs_X_x2__Y_simdvp__Z_w4,5,6_}
	\end{minipage}
\end{figure*}
\subsection{Semi-static Scheduling Policies}\label{subsec:offline-eval}
As discussed, semi-static scheduling policies in this work take the initial queue state at the moment of the arrival of a time-critical packet into account. %of importance
The semi-static schedule is then determined prior to the transmission of the corresponding packet and defines the allocation of slots until the deadline of the packet. 
We compare the DVP achieved by the proposed WTB-W scheduler, cf. Sec.~\ref{sec:offline-scheduling}, with the exhaustive search methods eDVPUB and Optimum, which compute semi-static policies, respectively, evaluating Eq.~\eqref{eq:opt-offline} and via simulations.
Due to its high complexity, the performances of the optimal policy is shown for smaller parameter sets.
Exhaustive schemes are shown to evaluate the performance of WTB-B and do not represent a feasible way of computing semi-static scheduling policies.
The last comparison scheme is a agnostic allocation of half of the slots for uplink and half of the slots for downlink transmission. 
We refer to this scheme as 50/50\footnote{\note{In the case of an odd number of slots, one extra slot is allocated to the first link.}}.
Throughout the evaluation of the semi-static schemes, we keep the channel error rate $p_e$ at $0.2$.

Fig.~\ref{fig:fig_lineplot_offline-high_xsymm_x1s_x2s_pes2_Ns4_ws_algs_X_w__Y_simdvp__Z_x1,1_2,2_3,3_} shows the DVP for different application deadlines, backlogs, while the frame configuration is fixed with $N$\,$=$\,$4$.
WTB-W and eDVPUB achieve close-to-optimal DVP for all backlog sizes and deadlines.
Thus, taking the initial queue states into account is beneficial and provides up to one order of magnitude improvement compared to the agnostic 50/50 scheme.
As the deadline $w$ increases, the performance gap of the proposed polices in comparison to the 50/50 scheme increases.
For backlogs $x_1$ and $x_2$ equaling $1$, eDVPUB provides a minor improvement in DVP with respect to WTB-W at the expense of higher computational complexity, while for higher backlogs, the performance difference between the two is negligible.

In Fig.~\ref{fig:fig_lineplot_offline-high_xsymm_x1s_x2s_pes2_Ns_ws5_algs_X_N__Y_simdvp__Z_x1,1_2,2_3,3_} we study the impact of different frame lengths, backlogs, while keeping the deadline fixed with $w$\,$=$\,$5$.
Again, the improvement in DVP of the proposed semi-static schedulers with respect to the agnostic 50/50 scheme is confirmed.
The ``stepped'' behaviour of the 50/50 scheme is caused by the different ratios of allocated slots to the two links for even and odd frame lengths.
Otherwise, the results show a minor difference between WTB-W and eDVPUB, while the optimality gap increases slightly for increasing $N$ when $x_1$ and $x_2$ are equal to $1$.

Fig.~\ref{fig:fig_lineplot_offline-high_x1inc_x1s_x2s_pes2_Ns4_ws_algs_X_x1__Y_simdvp__Z_w4,5,6_} and Fig.~\ref{fig:fig_lineplot_offline-high_x2inc_x1s_x2s_pes2_Ns4_ws_algs_X_x2__Y_simdvp__Z_w4,5,6_} show the impact of initial backlog on the DVP for different deadlines and a fixed frame configuration with $N$\,$=$\,$4$.
Increasing $x_1$ has a stronger impact than $x_2$ on the achievable DVP.
This is intuitive as packets backlogged in $x_1$ must be sent by both links.
We note that the gap between the proposed semi-static schedulers and the agnostic 50/50 scheme increases with an increasing backlog $x_2$.
%Thus, adopting the proposed scheduler is highly beneficial when the buffer of the second queue is larger.
In both figures, we observe again that exploiting the initial queue states leads to a gain of approximately half order of magnitude with respect to the agnostic 50/50 scheme with increasing initial backlogs.
%The optimality and performance gaps of the proposed schedulers is different for different initial backlogs at the queues.
As shown in Fig.~\ref{fig:fig_lineplot_offline-high_x1inc_x1s_x2s_pes2_Ns4_ws_algs_X_x1__Y_simdvp__Z_w4,5,6_}, eDVPUB achieves near-optimal DVP and the performance gap with respect to WTB-W is constant with increasing $x_1$.
In contrast, in Fig.~\ref{fig:fig_lineplot_offline-high_x2inc_x1s_x2s_pes2_Ns4_ws_algs_X_x2__Y_simdvp__Z_w4,5,6_}, the gap between the proposed schedulers and the optimal policy decreases with increasing $x_2$, achieving close-to-optimal DVP.

\begin{figure*}[!t]
	\centering
	\begin{minipage}{\columnwidth}
		\includegraphics{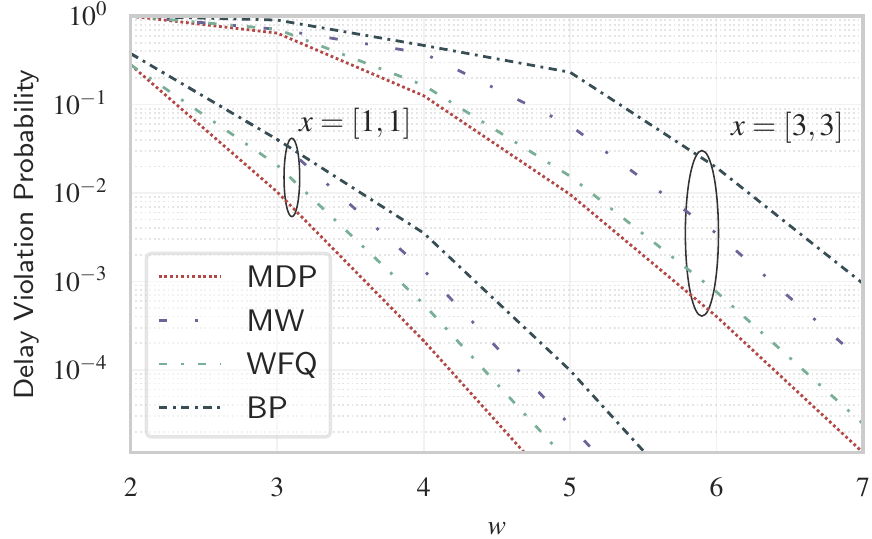}
		\caption{\note{DVP achieved by dynamic schedulers for different deadlines $w$, increasing backlogs $x_1, x_2$, $N=6$, $p_e=0.4$.}}
		\label{fig:fig_lineplot_online-high_xsymm_x1s_x2s_pes4_Ns6_ws_algs_X_w__Y_simdvp__Z_x1,1_3,3_}
	\end{minipage}\hfill
	\begin{minipage}{\columnwidth}
		\includegraphics{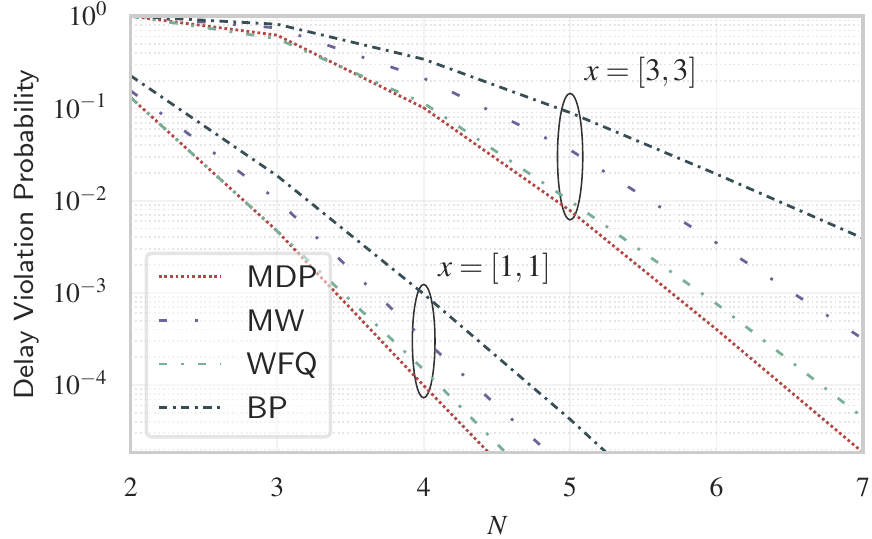}
		\caption{\note{DVP achieved by dynamic schedulers for different frame sizes $N$, increasing backlogs $x_1, x_2$, $w=6$, $p_e=0.4$.}}
		\label{fig:fig_lineplot_online-high_xsymm_x1s_x2s_pes4_Ns_ws6_algs_X_N__Y_simdvp__Z_x1,1_3,3_}
	\end{minipage}
\end{figure*}
\begin{figure*}[!t]
	\centering
	\begin{minipage}{\columnwidth}
		\includegraphics{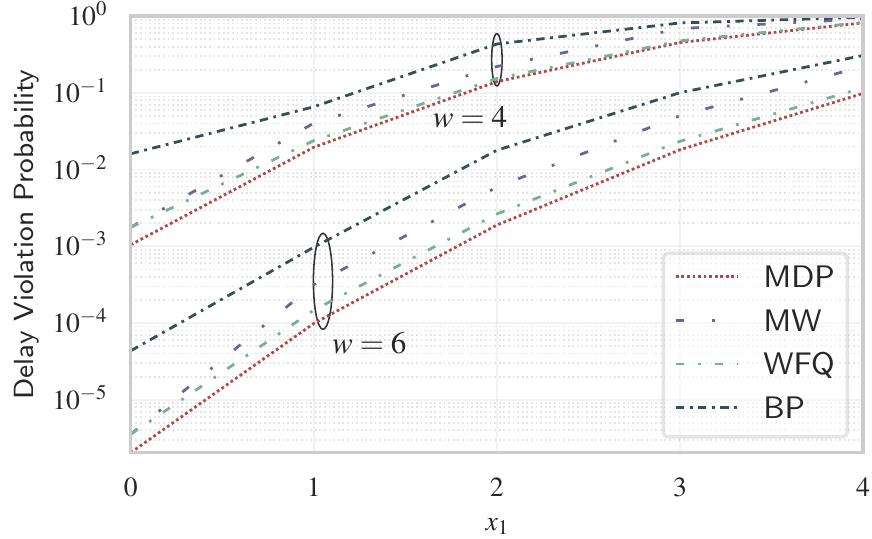}
		\caption{\note{DVP achieved by dynamic schedulers for different backlogs $x_1$ and deadlines $w$, $x_2=1$, $N=4$, $p_e=0.4$.}}
		\label{fig:fig_lineplot_online-high_x1inc_x1s_x2s_pes4_Ns4_ws_algs_X_x1__Y_simdvp__Z_w4,6_}
	\end{minipage}\hfill
	\begin{minipage}{\columnwidth}
		\includegraphics{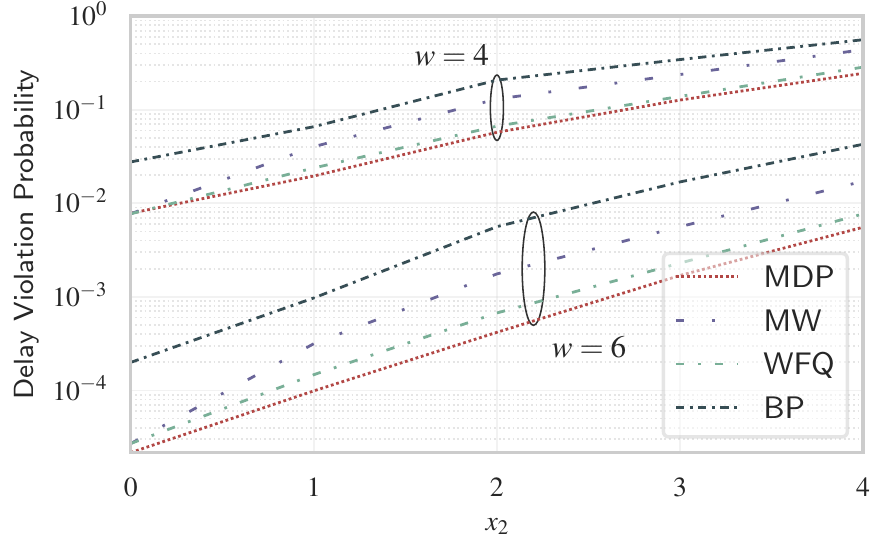}
		\caption{\note{DVP achieved by dynamic schedulers for different backlogs $x_2$ and deadlines $w$, $x_1=1$, $N=4$, $p_e=0.4$.}}
		\label{fig:fig_lineplot_online-high_x2inc_x1s_x2s_pes4_Ns4_ws_algs_X_x2__Y_simdvp__Z_w4,6_}
	\end{minipage}
\end{figure*}
%\vspace{-0.5cm}

\subsection{Dynamic Scheduling Policies}\label{subsec:online-eval}
In contrast to the semi-static schemes, dynamic schedulers benefit from feedback on the queue states as time evolves, giving them the opportunity to reallocate time slots depending on the queue's backlog evolution. 
This makes dynamic schedulers causing more overhead and complexity, with the advantage of potentially achieving a higher performance.
Initially, we are only turning to different dynamic schedulers in this section.
Concretely, we consider the following schemes:
\begin{itemize}
\item MDP: Our proposed dynamic scheduler from Section \ref{sec:online-scheduling}.
\item Max Weight (MW): Under MW, all slots are allocated to the link with the maximum backlog~\cite{neely2010stochastic}.
\item Weighted-Fair Queuing (WFQ): Under WFQ, slots are allocated to the links in proportion to the ratio between their queue sizes~\cite{parekh1993generalized}.
\item Backpressure (BP): Under BP, slots are allocated to the link with maximum backpressure, where the backpressure at the first link is equal to $x_1-x_2$ while at the second link it is equal to $x_2$~\cite{Tassiulas1990StabilityNetworks}.
\end{itemize}
In all the figures below, the value of $p_e$ is set to $0.4$.

% BP: boils down to difference btw the queue, give all the resources, in stationary sense, only few slots here (transient)
% MDP: expected throughput for exactly w slots, for long time maybe they conicide, links are correlated, 
% MDP vs BP: both look at state, but BP restricted to 1/0, MDP uses granular, BP
% WFQ vs MDP: both granular value
% BP vs MW: tries to give more importance to first queue
% while all other scenarios (BP) 
% consider links without dependancies,
% they have a 0/1 activation for links, we can activate both 
% classify where they pay off most
% large N
Fig.~\ref{fig:fig_lineplot_online-high_xsymm_x1s_x2s_pes4_Ns6_ws_algs_X_w__Y_simdvp__Z_x1,1_3,3_} shows the DVP achieved by the dynamic schedulers for different application deadlines, backlogs, and a fixed frame configuration with $N$\,$=$\,$6$.
We observe that the proposed MDP-based scheduler outperforms all the other methods.
BP and MW achieve higher DVP as they allocate all the slots to a single link in each frame.
Similar to MDP, WFQ allows a granular allocation of slots to the links and achieves the smallest performance gap.

Fig.~\ref{fig:fig_lineplot_online-high_xsymm_x1s_x2s_pes4_Ns_ws6_algs_X_N__Y_simdvp__Z_x1,1_3,3_} presents DVP achieved by different policies for different frame lengths, backlogs, and a fixed deadline $w$\,$=$\,$6$.
Again, MDP achieves lower DVP while the performances of the other schemes are in line with the previous scenario.
Furthermore, the performance advantage of MDP, with respect to the comparison schemes, increases with increasing frame length.
This advantage arises from the fact that the action space of MDP is larger for large $N$, which results in more accurate slot allocations.
% Furthermore, as discussed in~\cite{Zoppi2020dynamic}, increasing frame lengths increase the performance gap between MDP and the other dynamic schemes.
% This is explained by the fact that a bigger frame size increases the action space of the MDP, improving the evaluation of the value function for each system's state.

\begin{figure*}[!t]
	\centering
	\begin{minipage}{\columnwidth}
		\includegraphics{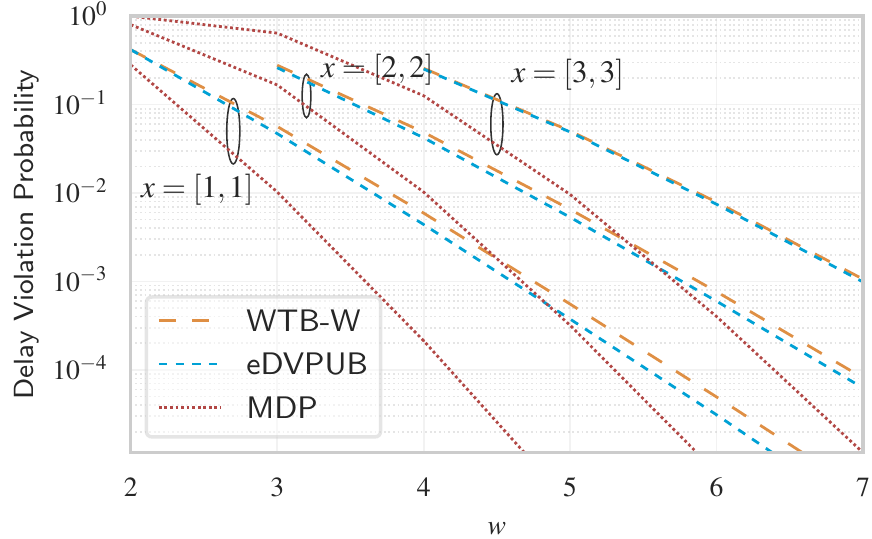}
		\caption{\note{DVP achieved by semi-static and dynamic schedulers for different deadlines $w$, increasing backlogs $x_1, x_2$, $N=6$, $p_e=0.4$.}}
		\label{fig:fig_lineplot_online-low_xsymm_x1s_x2s_pes4_Ns6_ws_algs_X_w__Y_simdvp__Z_x1,1_2,2_3,3_}
	\end{minipage}\hfill
	\begin{minipage}{\columnwidth}
		\includegraphics{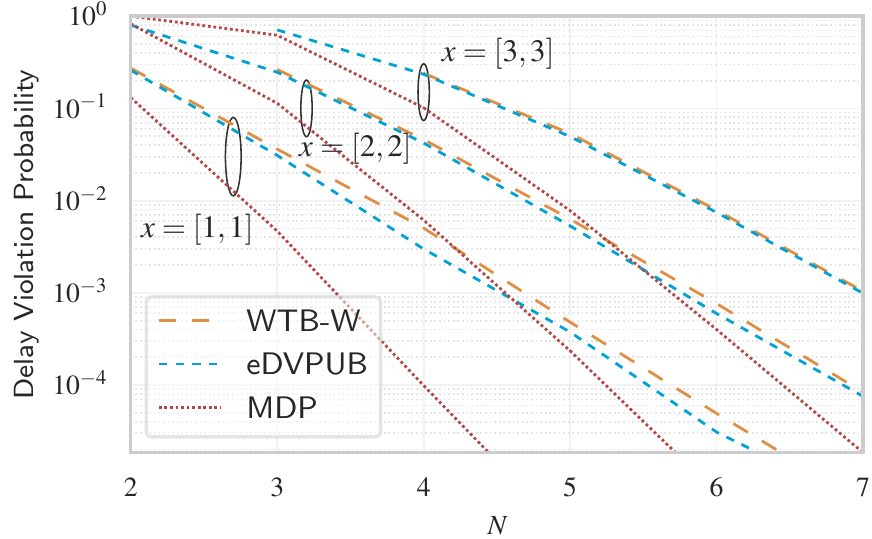}
		\caption{\note{DVP achieved by semi-static and dynamic schedulers for different frame sizes $N$, increasing backlogs $x_1, x_2$, $w=6$, $p_e=0.4$.}}
		\label{fig:fig_lineplot_online-low_xsymm_x1s_x2s_pes4_Ns_ws6_algs_X_N__Y_simdvp__Z_x1,1_2,2_3,3_}
	\end{minipage}
\end{figure*}
\begin{figure*}[!t]
	\centering
	\begin{minipage}{\columnwidth}
		\includegraphics{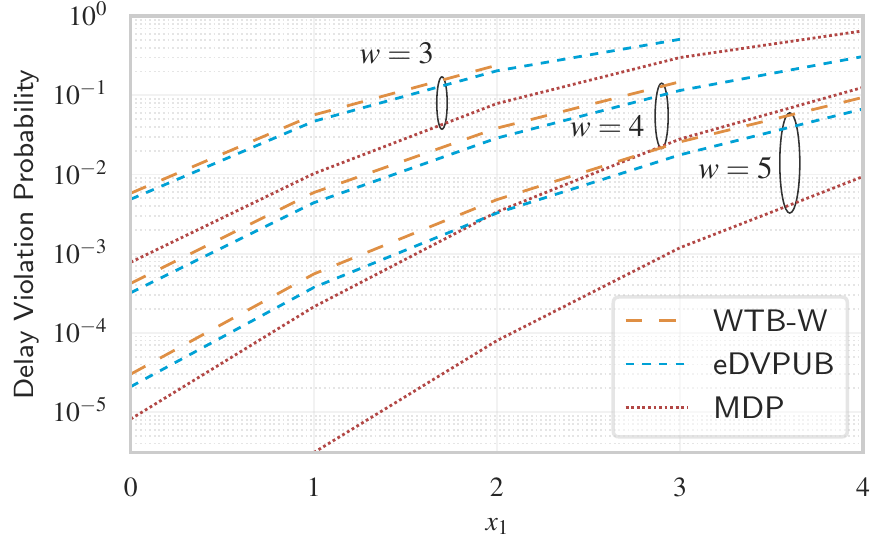}
		\caption{\note{DVP achieved by semi-static and dynamic schedulers for different backlogs $x_1$ and deadlines $w$, $x_2=1$, $N=6$, $p_e=0.4$.}}
		\label{fig:fig_lineplot_online-low_x1inc_x1s_x2s_pes4_Ns6_ws_algs_X_x1__Y_simdvp__Z_w3,4,5_}
	\end{minipage}\hfill
	\begin{minipage}{\columnwidth}
		\includegraphics{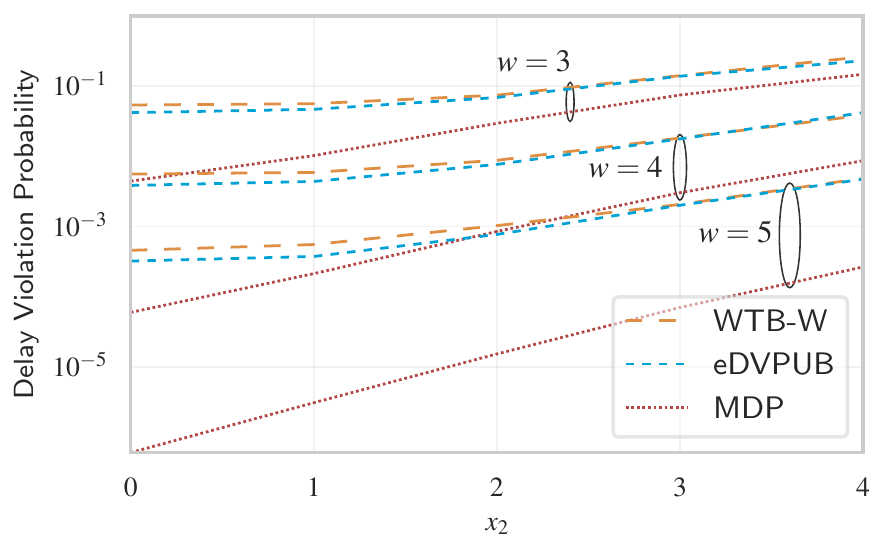}
		\caption{\note{DVP achieved by semi-static and dynamic schedulers for different backlogs $x_2$ and deadlines $w$, $x_1=1$, $N=6$, $p_e=0.4$.}}
		\label{fig:fig_lineplot_online-low_x2inc_x1s_x2s_pes4_Ns6_ws_algs_X_x2__Y_simdvp__Z_w3,4,5_}
	\end{minipage}
\end{figure*}
Fig.~\ref{fig:fig_lineplot_online-high_x1inc_x1s_x2s_pes4_Ns4_ws_algs_X_x1__Y_simdvp__Z_w4,6_} and Fig.~\ref{fig:fig_lineplot_online-high_x2inc_x1s_x2s_pes4_Ns4_ws_algs_X_x2__Y_simdvp__Z_w4,6_} show the impact of the initial backlogs on the DVP achieved by different policies for different deadlines and with a fixed frame configuration $N$\,$=$\,$4$.
Again, increasing $x_1$ has a higher impact on DVP compared to $x_2$.
Due to the allocation of all slots to a link in a frame, the performance gap of BP and MW increases as $x_1$ and $x_2$ increase. 
WFQ, however, maintains a constant gap, being able to adapt the allocations to different initial backlog scenarios.

\subsection{Impact of Network State Information}\label{subsec:off-on-eval}
A direct comparison of the proposed semi-static and dynamic scheduling policies allows to quantify the performance improvement achieved by exploiting up-to-date queue states.
In the following we limit this comparison to the proposed schemes of this paper, WTB-W and MDP, and additionally show eDVPUB to represent the close-to-optimal performances of semi-static policies.
Therefore, the semi-static schemes WTB-W and eDVPUB are benchmarked with the dynamic MDP scheme.
In all following figures, the value of $p_e$ is set to $0.4$.

In Fig.~\ref{fig:fig_lineplot_online-low_xsymm_x1s_x2s_pes4_Ns6_ws_algs_X_w__Y_simdvp__Z_x1,1_2,2_3,3_} the DVP achieved by the proposed scheduling policies is presented for different deadlines, backlogs, and fixed frame configuration $N$\,$=$\,$4$.
Most importantly, we witness a significant performance advantage of MDP in comparison to WTB-W and eDVPUP.
This advantage increases with increasing deadlines, reaching multiple orders of magnitudes.
This effect is intuitive as, at each frame, dynamic scheduling benefits from up-to-date queue states.
A similar effect can be observed in Fig.~\ref{fig:fig_lineplot_online-low_xsymm_x1s_x2s_pes4_Ns_ws6_algs_X_N__Y_simdvp__Z_x1,1_2,2_3,3_} where the DVP performance is shown for different frame lengths, backlogs, and a fixed deadline $w$\,$=$\,$6$. 
%Results confirm that additional queue states are exploited by dynamic policies to achieve lower DVP compared to semi-static policies.
Finally, Fig.~\ref{fig:fig_lineplot_online-low_x1inc_x1s_x2s_pes4_Ns6_ws_algs_X_x1__Y_simdvp__Z_w3,4,5_} and~\ref{fig:fig_lineplot_online-low_x2inc_x1s_x2s_pes4_Ns6_ws_algs_X_x2__Y_simdvp__Z_w3,4,5_} show the impact of initial backlog on the DVP for different deadlines and a fixed frame configuration $N$\,$=$\,$6$.
For fixed system parameters, increasing the initial backlogs results in an increasing DVP, which translates into smaller performance gaps between semi-static and dynamic scheduling schemes.

\note{
\section{Conclusions}\label{sec:conclusions}
In this work, we investigated semi-static and dynamic scheduling policies to support the feedback-based communication of interactive applications.
In particular, we proposed scheduling policies aimed at minimizing the end-to-end latency experienced by individual messages traversing the feedback loop, i.e. the delay from capturing the status information until the point in time when the corresponding feedback information is exposed.
The proposed scheduling policies allocate time slots to the transmitters in order to minimize the delay violation probability (DVP) taking the network queue states into account.
% The operation of interactive applications over a communication network requires strict quality-of-service for the communication of automated feedback messages. 
% To this end, we analysed the end-to-end latency of the entire feedback loop, i.e. the delay from capturing the status information until the point in time when the corresponding feedback information is exposed.
% With the aim of optimizing the network infrastructure to support such applications, we investigated network resource management strategies that minimize the delay violation probability (DVP), which is defined as the tail probability of the end-to-end latency exceeding a target value.
 %We characterized the DVP of a two-hop network representing the uplink/downlink communication of a feedback loop.
%and investigated scheduling methods that alloc in competition to the two transmitters.
%Considering a time slotted medium access where time slots are organized in frames, the proposed scheduling policies allocate time slots to the two transmitters that compete for resources.
%In particular, noting that the initial queue states strongly influence the end-to-end latency of feedback messages, we investigated scheduling policies that include this information in two different ways.
On the one hand, \emph{semi-static} policies compute schedules based on the initial queue states by efficiently minimizing the WTB, which is derived applying the union and Chernoff bounds.
%$Noting that the closed-form expression of DVP is intractable, we developed heuristic semi-static scheduling policies that are based on two upper bounds on DVP.
On the other hand, \emph{dynamic} policies solve a throughput-optimal MDP, which allocates time slots depending on the queue’s backlog evolution.
%Noting that DVP cannot be directly used for dynamic resource allocation, we derived a dynamic heuristic scheduling policy that maximizes the network's throughput.
Via simulations of several main system parameters and comparison baselines, we demonstrate the effectiveness of the proposed methods in reducing DVP. 
%We evaluated the performance of the proposed scheduling policies numerically .
The results show that the proposed WTB-W semi-static scheduling policy achieves up to one order of magnitude improvement compared to a queue-agnostic scheme and close-to-optimal DVP.
Furthermore, simulations prove the superiority of the proposed MDP dynamic scheduler in reducing DVP compared to the existing Backpressure, Max Weight, and Weighted-Fair Queuing algorithms.
%MDP outperforms all the other scheduling methods and its potential is fully exploited for large frame lengths.
Finally, we quantified the performance improvement achieved by exploiting up-to-date queue states with a direct comparison of the proposed semi-static and dynamic scheduling policies.
MDP achieves a significant performance advantage in comparison to WTB-W reaching multiple orders of magnitude.

The contributions of this paper leave space for interesting future work.
The proposed scheduling policies can be investigated for more complex systems, for instance taking into account flow transformation within the feedback loop,  asymmetric transmitter PERs, and non-stationary link qualities. 
Furthermore, the problem can be extended considering multiple feedback loops sharing the same network, thus minimizing the DVP of packets from multiple applications.  
Finally, an interesting next step is to investigate the application of the proposed scheduling policies in the existing wireless systems, such as NCS operating in IWSN or in 5G cellular networks.
}

\appendix
\note{\subsection{Proof of Proposition~\ref{prop1}}\label{app:dvp-derivation}
\allowdisplaybreaks
Combining \eqref{eq:dvp-def} and~\eqref{eq:dynamic-server}, the DVP can be calculated as
\begin{IEEEeqnarray}{rCl}
	\IEEEeqnarraymulticol{3}{l}{
		\text{DVP}(w,y,x_1,x_2) =
	}\nonumber\\* \quad
	&=&\P\{\D(w) < y + x_1 + x_2\} \nonumber\\
	&=& \P\left\{\underset{0\leq u \leq w}{\min} \left[\S^2(w-u) + \A^2(u)<y + x_1 + x_2\right]\right\} \nonumber\\
	&=& \P\Big\{\underset{0\leq u \leq w}{\min} \left[\S^2(w-u) + \D^1(u-1) + x_2\right]\nonumber\\
	&& \qquad\qquad\qquad<y + x_1 + x_2\Big\}\nonumber\\
%\end{IEEEeqnarray}
%\begin{IEEEeqnarray}{rCl}
	&=& \P\Big\{ \!\left\{\S^2(w) \! < \! y + x_1 + x_2\right\} \cup \left\{\S^2(w-1) \! < \! y + x_1\right\} \cup \nonumber\\
	&&\quad\bigcup_{u=2}^{w}\big\{\underset{0\leq v \leq u-1}{\min} [\S^2(w-u) + \S^1(u-1-v) + \nonumber\\
	&&\qquad\qquad\qquad\A(v) < y + x_1]\big\} \Big\}\nonumber\\
%\end{IEEEeqnarray}
%\begin{IEEEeqnarray}{rCl}
	&=& \P\Big\{ \!\left\{\S^2(w) \! < \!y + x_1 + x_2\right\} \cup \left\{\S^2(w-1) \! < \! y + x_1\right\} \cup \nonumber\\
	&&\quad\bigcup_{u=2}^{w} \left\{\S^2(w-u) + \S^1(u-1) < y + x_1\right\}\Big\}. \nonumber%\label{eq:exact-dvp2q}. 
\end{IEEEeqnarray}

\subsection{Proof of Theorem~\ref{th1}}\label{app:wtb2q-convexity}
% \begin{theorem}\label{th:convexity-dynamic-all}
% 	The WTB is convex $\forall n, s>0$.
% \end{theorem}
% \begin{proof}
To simplify the analysis of Eq.~\eqref{eq:wtb-w-server}, we use $\alpha = (1-p_e) e^{-s}+p_e$ and $\beta=k_a+x_1-1$

%\begin{figure*}[!t]
%\normalsize
%\setcounter{tempeqcounter}{\value{equation}} % temp store of current value
\begin{IEEEeqnarray}{rCl}
    \text{WTB}(\mathbf{n}^2) &=& \min_{s>0} \alpha^{\sum_{i=0}^{w} n^i} e^{s\,(\beta+x_2)} +\nonumber\\
    &&\sum_{u=1}^{1+w}\alpha^{(u-1)N-\sum_{j=0}^{u-2} n^j+\sum_{k=0}^{w-u} n^k} e^{s\,\beta}\nonumber\\
    &=& \min_{s>0} e^{s\,\beta} \Bigg[\alpha^{\sum_{i=0}^{w} n^i} e^{s x_2} +\nonumber\\
    &&\sum_{u=1}^{1+w} \alpha^{(u-1)N-\sum_{j=0}^{u-2} n^j+\sum_{k=0}^{w-u} n^k}\Bigg].\label{eq:wtb2q-d-rewritten}
\end{IEEEeqnarray}
%\setcounter{equation}{\value{tempeqcounter}} % restore correct value
%\hrulefill
%\vspace*{4pt}
%\end{figure*}

% Convexity definition
% \begin{IEEEeqnarray}{rCl}
%     &&\text{WTB2Q-D}\left(\lambda \mathbf{n}_1 + (1-\lambda)\mathbf{n}_2\right) \leq \lambda\text{WTB2Q-D}\left(\mathbf{n}_1\right) + (1-\lambda)\text{WTB2Q-D}\left(\mathbf{n}_2\right).
% \end{IEEEeqnarray}
Applying the definition of convexity to Eq.~\eqref{eq:wtb2q-d-rewritten}, $\forall \mathbf{n},\mathbf{m} \in\mathbb{N}^{1 \times w}, \lambda\in\left[0,1\right]$ we obtain Eq.~\eqref{eq:wtb2q-d-proof}.
\addtocounter{equation}{1}%
\setcounter{storeeqcounter}%
{\value{equation}}%
From Eq.~\eqref{eq:wtb2q-d-proof}, (1) we have applied the definition of convexity to the exponentials $\alpha^{f(\lambda\mathbf{n^1}+(1-\lambda)\mathbf{n^2})}$ knowing that they are convex and (2) we have used the fact that the minimum of the sum is smaller or equal of the sum of the minima.}
% \end{proof}

\begin{figure*}[!t]
\normalsize
\setcounter{tempeqcounter}{\value{equation}} % temp store of current value
\note{\begin{IEEEeqnarray}{rCl}
\setcounter{equation}{\value{storeeqcounter}} % number of this equation
    &&\text{WTB}\left(\lambda \mathbf{n} + (1-\lambda)\mathbf{m}\right)\nonumber\\
    &&=\min_{s>0} e^{s\,\beta} \left[\alpha^{\sum_{i=0}^{w}\lambda n^1_i + (1-\lambda)n^2_i} e^{s x_2} + \sum_{u=1}^{1+w} \alpha^{(u-1)N-\sum_{j=0}^{u-2} \lambda n^1_j + (1-\lambda)n^2_j+\sum_{k=0}^{w-u} \lambda n^1_k + (1-\lambda)n^2_k}\right] \nonumber\\
    &&\symtext{(1)}{\leq}\min_{s>0} e^{s\,\beta} \left[e^{s x_2}\left(\lambda\alpha^{\sum_{i=0}^{w} n^1_i}+(1-\lambda)\alpha^{\sum_{i=0}^{w}n^2_i} \right) + \right.\nonumber\\
    &&\left.\qquad\qquad\quad\sum_{u=1}^{1+w} \alpha^{(u-1)N} \left(\lambda\alpha^{-\sum_{j=0}^{u-2} n^1_j} + (1-\lambda)\alpha^{-\sum_{j=0}^{u-2}n^2_j}+\lambda \alpha^{\sum_{k=0}^{w-u} n^1_k} + (1-\lambda)\alpha^{\sum_{k=0}^{w-u} n^2_k}\right)\right] \nonumber\\
    &&=\min_{s>0} e^{s\,\beta} \left[\lambda\left(\alpha^{\sum_{i=0}^{w} n^1_i}e^{s x_2}+\sum_{u=1}^{1+w} \alpha^{(u-1)N} \alpha^{-\sum_{j=0}^{u-2} n^1_j} \alpha^{-\sum_{j=0}^{u-2} n^1_j}\right) + \right.\nonumber\\
    &&\left.\qquad\qquad\quad(1-\lambda)\left(\alpha^{\sum_{i=0}^{w} n^2_i}e^{s x_2}+\sum_{u=1}^{1+w} \alpha^{(u-1)N} \alpha^{-\sum_{j=0}^{u-2} n^2_j} \alpha^{-\sum_{j=0}^{u-2} n^2_j}\right)\right]\nonumber\\
    &&\symtext{(2)}{\leq}\lambda\min_{s>0}e^{s\,\beta}\left[ \alpha^{\sum_{i=0}^{w} n^1_i}e^{s x_2}+\sum_{u=1}^{1+w} \alpha^{(u-1)N} \alpha^{-\sum_{j=0}^{u-2} n^1_j} \alpha^{-\sum_{j=0}^{u-2} n^1_j}\right]+\nonumber\\
    &&\quad(1-\lambda)\min_{s>0}e^{s\,\beta}\left[ \alpha^{\sum_{i=0}^{w} n^2_i}e^{s x_2}+\sum_{u=1}^{1+w} \alpha^{(u-1)N} \alpha^{-\sum_{j=0}^{u-2} n^2_j} \alpha^{-\sum_{j=0}^{u-2} n^2_j}\right]\nonumber\\
    &&=\lambda\text{WTB}\left(\mathbf{n}\right) + (1-\lambda)\text{WTB}\left(\mathbf{m}\right).\label{eq:wtb2q-d-proof}
\end{IEEEeqnarray}}
\setcounter{equation}{\value{tempeqcounter}} % restore correct value
\hrulefill
\vspace*{4pt}
\end{figure*}

\section*{Acknowledgment}
This work was supported by the DFG Priority Programme 1914 grant number KE1863/5-1.

\bibliographystyle{IEEEtran}
%\bibliography{./resources/references,./resources/sota-icc,./resources/mendeley}
\bibliography{./resources/references-all}

%\newpage
%\input{./sections/results-discussion.tex}

\end{document}